\documentclass[12pt,a4paper,reqno]{amsart}
\usepackage{amsmath,amssymb,graphics,epsfig,color,enumerate,mathrsfs}
\usepackage{dsfont}  \usepackage{verbatim}
\definecolor{refkey}{gray}{.78}
\definecolor{labelkey}{gray}{.78}

\newtheorem{Theorem}{Theorem}[section]

\newtheorem{TheoremA}{Theorem}
\newtheorem{Lemma}[Theorem]{Lemma}
\newtheorem{Proposition}[Theorem]{Proposition}
\newtheorem{Corollary}[Theorem]{Corollary}
\newtheorem{Remark}[Theorem]{Remark}

\newtheorem{Claim}[Theorem]{Claim}
\newtheorem{Definition}[Theorem]{Definition}
\newtheorem{Warning}{Warning}[section]

 \definecolor{darkgreen}{rgb}{0,0.4,0}

\definecolor{light}{gray}{0.9}

\newcommand{\cA}{\ensuremath{\mathcal A}}
\newcommand{\cB}{\ensuremath{\mathcal B}}
\newcommand{\cC}{\ensuremath{\mathcal C}}

\newcommand{\cE}{\ensuremath{\mathcal E}}

\newcommand{\cG}{\ensuremath{\mathcal G}}
\newcommand{\cH}{\ensuremath{\mathcal H}}

\newcommand{\cN}{\ensuremath{\mathcal N}}

\newcommand{\cP}{\ensuremath{\mathcal P}}

\newcommand{\cW}{\ensuremath{\mathcal W}}

\newcommand{\bbB}{{\ensuremath{\mathbb B}} }

\newcommand{\bbE}{{\ensuremath{\mathbb E}} }

\newcommand{\bbL}{{\ensuremath{\mathbb L}} }

\newcommand{\bbN}{{\ensuremath{\mathbb N}} }

\newcommand{\bbR}{{\ensuremath{\mathbb R}} }

\newcommand{\bbZ}{{\ensuremath{\mathbb Z}} }

\let\a=\alpha \let\b=\beta   \let\d=\delta  \let\e=\varepsilon
 \let\g=\gamma     \let\k=\kappa  \let\l=\lambda
      \let\o=\omega      
  \let\s=\sigma \let\t=\tau   
  
\let\D=\Delta     \let\L=\Lambda 
\let\O=\Omega      

\newcommand{\da}{\downarrow}

\newcommand{\toup}{\rightharpoonup}

\newcommand{\be}{\begin{equation}}
\newcommand{\en}{\end{equation}}
\newcommand{\bes}{\begin{equation*}}
\newcommand{\ens}{\end{equation*}}

\thanks{This work has been partially supported by the ERC Starting Grant 680275 MALIG and by the Grant  PRIN 20155PAWZB "Large Scale Random Structures"}

\author[A.~Faggionato]{Alessandra Faggionato}
\address{Alessandra Faggionato.
  Dipartimento di Matematica, Universit\`a di Roma `La Sapienza'
  P.le Aldo Moro 2, 00185 Roma, Italy}
\email{faggiona@mat.uniroma1.it}

\newcommand{\ra}{\rangle}
\newcommand{\la}{\langle}

\title[Miller--Abrahams resistor network and  Mott random walk]{Miller--Abrahams random  resistor network, Mott random walk and 2-scale homogenization}
\begin{document}

\maketitle

\textcolor{blue}{This preprint is  and will remain unpublished.  The results contained here have been generalized to a very large class of random resistor networks in arXiv:2108.11258: A. Faggionato; \emph{Scaling limit of the conductivity of random resistor networks on point processes}.  Please see arXiv:2108.11258 for results and proofs.
}

\begin{abstract} The Miller-Abrahams (MA) random resistor network is given by a complete graph on a marked simple  point process with  edge conductivities depending on the marks and decaying exponentially in the edge length. As  Mott random walk,  it is an  effective model to study Mott variable range hopping in amorphous solids as doped semiconductors. By using 2-scale homogenization 
we prove that a.s.  the  infinite volume conductivity  of the MA  resistor network  is given by  an effective  homogenized matrix $D$ times the mean point density. We also derive homogenization results for the electrical potential.
Moreover, $D$ admits a variational characterization and  equals  the limiting diffusion matrix of Mott   random walk. This result clarifies the relation between the two models and it  also  allows to extend to the MA  resistor network   the existing bounds on  $D$ in agreement with the physical Mott law \cite{FM,FSS}. The latter   
concerns the low temperature stretched exponential decay of conductivity in amorphous solids. The techniques developed here  can be applied to other models, as e.g. the random conductance model, without ellipticity assumptions.

\smallskip

\noindent 
{\em Keywords}:
Marked simple  point process,   Mott variable range hopping, Miller--Abrahams random resistor network,  Mott random walk, homogenization, 2-scale convergence.

\smallskip

\noindent
{\em AMS 2010 Subject Classification}: 
60G55, 
74Q05, 
82D30 

\end{abstract}

\section{Introduction} The  Miller--Abrahams (MA) random resistor network  \cite{MA} has been introduced in order to study   electron transport in amorphous media as doped semiconductors in the regime of strong Anderson localization. These solids present an anomalous conductivity decay at zero temperature, described by Mott law. 
Calling $x_i$ the  impurity positions  in the doped semiconductor, the electron Hamiltonian has exponentially localized quantum eigenstates with localization centers $x_i$ and corresponding energy $E_i$  close to the Fermi level, set equal to zero in what follows. At low temperature phonons  induce transitions between the localized eigenstates, the rate of which can be calculated by the Fermi golden rule \cite{MA,SE}. In the simplification of spinless electrons,   the resulting rate for an electron to  hop from $x_i$ to the unoccupied site $x_j$ is then given by  (cf. \cite[Eq.~(3.6)]{AHL}) 
\be\label{tunnel} \exp\Big\{-\frac{2}{\g}  |x_i-x_j| - \beta \{ E_j -E_i \}_+\Big\} \,.\en
In \eqref{tunnel} $\g$ is the localization length, $\beta=1/k T $ is the inverse temperature and $\{a\}_+:= \max\{0,a\}$.

The above set $\{x_i\}$ can be modelled by a   random simple point process, marked by  random variables $E_i$ (called energy marks) which can be taken i.i.d. with common distribution $\nu$. The physically relevant distributions for inorganic media are of the form  $\nu(dE)= c\, |E|^{\alpha} dE$ with  finite support  $[-A,A]$  for some exponent $\a\geq 0$ \cite{Min,SE} (one says that the marked simple point process $\{(x_i,E_i)\}$ is the $\nu$--randomization of $\{x_i\}$). Mott law \cite{Mott,MD,SE} then  predicts  that, for $d\geq 2$,  the DC conductivity matrix $\s(\b)$ of the medium decays to zero as  $\b \nearrow \infty$ as 
\be\label{mottino}
\s(\b) \approx  A(\b)\exp \bigl\{ -\k \, \beta ^\frac{\a+1}{\a+d +1}\bigr\}\,,
 \en
 where the prefactor matrix $A(\b)$ exhibits  a negligible $\b$--dependence (we keep the matrix  formalism to cover  anisotropic media).
 Strictly speaking, Mott derived the above asymptotics for $\a=0$, while  Efros and Shklovskii derived it for $\a=d-1$. The cases $\a=0,d-1$ are the two physically relevant ones. 
 For $d=1$ the DC conductivity presents an Arrenhius--type decay  as $\b \nearrow \infty$  for all $\a\geq 0$ \cite{Kur}, i.e.\
 \be\label{mottino1}
 \s(\b)\approx  A(\b)\exp \bigl\{ -\k \b\}\,.
 \en

Due to   localization one can treat the above electron conduction by a hopping process of classical particles (see also \cite{ABS,BRSW}), thus leading  anyway to a complicate simple exclusion process due to the  Pauli blocking. The reversible measure of the exclusion process is  the Fermi-Dirac distribution (i.e. the Bernoulli product probability measure  such that the probability of having a particle at $x_i$ is proportional to $e^{-\b (E_i-\g)}$, $\g\in \bbR$ being the chemical potential)  .
 Effective simplified models in  a  mean field approximation   are  given by the MA random resistor network \cite{AHL,MA,POF,SE}   and  by Mott random walk \cite{FSS}.
The MA random resistor network has nodes $x_i$ and,   between any  pair of nodes $x_i\not=x_j$, it has  an electrical filament of  conductivity 
\be\label{condu}
c_{x_i,x_j}:=\exp\Big\{ - \frac{2}{\gamma} |x_i-x_j| -\frac{\b}{2} ( |E_i|+ |E_j|+ |E_i-E_j|) \Big\}\,.
\en
Mott random walk is the continuous--time random walk with state space $\{x_i\}$ and probability rate for a jump from $x_i$ to $x_j$ given by \eqref{condu}. We point out that the r.h.s. of \eqref{condu}   corresponds to the leading term for $\b$ large of \eqref{tunnel} multiplied by the probability in the Fermi-Dirac distribution that $x_i$ and $x_j$ are, respectively,   occupied and unoccupied  by an electron 
(see \cite[Eq.~(3.7)]{AHL}). 

\smallskip

The original  derivation of the  laws \eqref{mottino} and \eqref{mottino1} is rather heuristic.  More robust    arguments  have been proposed in the physical  literature (see \cite{AHL,MA,POF,Sa,SE}). We recall   some rigorous results  for Mott random walk.  They hold under general conditions (see the references below for the details).
 We start with $d\geq 2$.
In \cite[Thm.~1]{FSS} and  \cite[Thm.~1.2]{CFP} an invariance principle (respectively annealed and quenched) is stated  for Mott random walk, with asymptotic diffusion matrix $D(\beta)$ admitting a variational characterization \cite[Thm.~2]{FSS}.  In addition,    lower and upper bounds on $D(\b)$ in agreement with Mott law  \eqref{mottino} have been obtained respectively in \cite[Thm.~1]{FSS} and \cite[Thm.~1]{FM}: for $\b$ large
\be \label{silenzioso1}
c_1 \exp \bigl\{ -c_1' \, \beta ^\frac{\a+1}{\a+d +1}\bigr\}\mathbb{I} \leq D(\b)\leq 
c_2 \exp \bigl\{ -c_2' \, \beta ^\frac{\a+1}{\a+d +1}\bigr\}\mathbb{I}\,,
\en
for suitable $\b$--independent positive  constants $c_1,c_1',c_2,c_2'$.
For $d= 1$  annealed and quenched  invariance principles  have been obtained in \cite[Thm.~1.1]{CP1}.  Again $D(\b)$ has a variational characterization and satisfies bounds in agreement with \eqref{mottino1} (see \cite[Thm.~1.2]{CP1}):
\be \label{silenzioso2}
c_1 \exp \bigl\{ -c_1' \, \beta  \bigr\} \leq D(\b)\leq 
c_2 \exp \bigl\{ -c_2' \, \beta \}\,,
\en
for suitable $\b$--independent positive  constants $c_1,c_1',c_2,c_2'$.
By invoking Einstein relation (which has been rigorously proved for $d=1$ in \cite{FGS2}) the bounds in \eqref{silenzioso1} and \eqref{silenzioso2} extend to the mobility matrix defined in terms of linear response.

\smallskip
Similar results for the  conductivity matrix  of the  MA resistor network  were absent.
Our main result (cf. Theorem \ref{teo1})  fills this gap and clarifies the connection between the MA resistor network  and Mott random walk. Indeed, for  ergodic stationary  marked simple point processes  $\{(x_i,E_i)\} $ we prove  that the  limit of the conductivity of the MA resistor network read on enlarging boxes exists  under suitable rescaling (we call this  limit the \emph{infinite volume conductivity matrix}). We also provide a variational characterization of the infinite volume conductivity matrix of the  MA resistor network, implying that 
 it   equals  the asymptotic diffusion matrix $D(\beta)$ of Mott random walk times the mean point density.
 As a consequence we get that   the infinite volume conductivity matrix satisfies \eqref{silenzioso1} for $d\geq 2$ under the assumptions of \cite[Thm.~1]{FSS} and \cite[Thm.~1]{FM} and satisfies 
 \eqref{silenzioso2} for $d=1$ under the assumptions  of \cite[Thm.~1.2]{CP1}.  The matrix $D(\b)$ equals  also the effective homogenized matrix associated to the  rescaled Markov generator of Mott random walk (see \cite[Thm.~1]{Fhom}).
 As a consequence, due to \cite[Thm.~2]{Fhom}, under conditions much weaker than the ones leading to the above quenched/annealed invariance principles, Mott random walk satisfies a weak form of central limit theorem with asymptotic diffusion matrix $D(\b)$.
 Our second main result  is given by the homogenization property of the electrical potential in the MA random resistor network (cf. Theorem \ref{teo2}).  
 
We point out that our results do not restrict to Mott variable range hopping (shortly, v.r.h.), i.e. to the MA random resistor network with conductivities \eqref{condu}. Indeed, our  Theorems \ref{teo1} and \ref{teo2} are stated  for more general MA random resistor networks.
We also stress  that we have followed here the convention  used  in Physics   for the diffusion matrix, which is  given by twice the covariance matrix of $B_1$, where  $(B_t)_{t\geq  0}$ is  the  Brownian motion emerging in the CLT/invariance principle. This explains the  factor  $1/2$
  appearing in  Definition \ref{fasma98} for $D$  and not appearing in \cite[Thm.~2]{FSS}, \cite[Thm.~1.1]{CP1}. This choice has the advantage to identify $D(\b)$ with the effective homogenized matrix (cf. \cite{Fhom}).


 \smallskip
 
We conclude with some  comments on the technical aspects. Our  proofs are based on homogenization with 2-scale convergence (cf. \cite{Z,ZP} and references therein). Thinking of $\o:=\{(x_i,E_i)\}$ as a microscopic picture of the medium   and introducing the scaling parameter $\e>0$, the 2-scale convergence allows  to
explore the  the ergodicity  properties  of the medium (cf. Prop. \ref{prop_ergodico} below) when  averaging on enlarging boxes of size $1/\e$ quantities as $\varphi(\e x_i) g(\t_{x_i} \o)$,    $\t_{x_i} \o$ being the environment viewed from site $x_i$.  Note that, while $\e x_i $ belongs to the macroscopic world, $\t_{x_i} \o$ refers to the microscopic one (hence the presence of 2 scales).

 In \cite{ZP} the authors have proved homogenization  for the Poisson equation $u+ \bbL u=f$ by 2-scale convergence, on $\bbR^d$ and on bounded domains with mixed boundary conditions, $\bbL$ being the generator of a  diffusion in random environments. Analogous results for Mott random walk  on $\bbR^d$, but not on bounded domains,  have been obtained in \cite{Fhom}. In  \cite[Section 7]{ZP}  the above  results of \cite{ZP} on bounded domains have been applied   to get that the effective homogenized  matrix  $D$ equals the infinite volume conductivity in a  model related to percolation, under the  a priori check  that $D>0$.  
 To get Theorem \ref{teo1} one could  have also thought  to adapt the strategy  
 developed for diffusions with random coefficients  in  \cite{BP} to difference operators by using the results of \cite{PR},  but \cite{PR} requires ellipticity assumptions (which are not valid in Mott v.r.h.).
    We have developed here a direct proof based on 2-scale homogenization, which  avoids both the  a priori check that $D>0$ and  elliptic assumptions.     Our proof of Theorem \ref{teo1} and \ref{teo2}  is very general and can be applied as well to other resistor networks, as e.g. the  conductance model \cite{Bi}  without any ellipticity assumption. For the conductance model the identification between the asymptotic diffusion matrix and  the conductivity matrix had already been derived under ellipticity  assumptions (see \cite{BSW} and references therein).  We stress that, in addition to the lack of ellipticity, Mott v.r.h. presents further technical difficulties due to long  jumps and the absence of an underlying lattice structure (thus leading to the concept of amorphous gradients), not present in the above models.
      
    We conclude mentioning that other rigorous results 
   on the  Miller--Abrahams random resistor network 
have been recently obtained in \cite{FAMH1} and \cite{FAMH2}, where percolation properties of the subnetwork given by filaments with lower bounded conductances  have been analyzed.
By means of  these results and the present Theorem \ref{teo1}, 
in a forthcoming paper  we will show for $d\geq 2$  that one can go beyond the bounds \eqref{silenzioso1} and get the asymptotics of the infinite volume conductivity for $\b$ large.

%
\medskip
\noindent
{\bf Outline of the paper}. In the rest  we remove the inverse temperature $\b$ from the notation. In Section \ref{MM} we introduce the model and state our main results (cf.  Theorem \ref{teo1} and  Theorem \ref{teo2} for $D_{1,1}>0$). In Section \ref{eff_equation} we analyze the effective diffusive equation. In Section \ref{GS} we recall  basic facts on marked simple point processes and their Palm distribution. In Section \ref{sec_hilbert} we introduce the proper Hilbert space to  analyze the electrical potential. In Section \ref{renato_zero} we prove Theorem \ref{teo1} when $D_{1,1}=0$. In Section \ref{sec_review} we consider the space of square integrable forms. In Section \ref{sec_tipetto} we define the family of typical environments. In section \ref{anatre12} we recall the definitions of several types of convergence (including the weak 2-scale convergence). Sections \ref{sec_bike} and \ref{limit_points} are devoted to the weak 2-scale limit  points of the electrical potential and its gradient. Finally, in Section \ref{limitone} we conclude the proof of Theorems \ref{teo1} and \ref{teo2} when $D_{1,1}>0$. We collect some minor results in Appendix \ref{ailo}.

\section{Model and main results}\label{MM}
We 
 denote by  
   $\O$ the space of  locally finite subsets $\o \subset \bbR^d\times \bbR$ such that for each $x\in \bbR^d$ there exists at most one element $E\in \bbR$ with $(x,E)\in \o$.    
We write a generic element  $\o \in \O$ as 
 $\o= \{ ( x_i , E_i)\}$ ($E_i$ is called the  \emph{mark} at the point  $x_i$) and we set  $\hat\o :=\{x_i\}$.  
 We will identify the sets $\o= \{ ( x_i , E_i)\}$ and $\hat\o =\{x_i\}$ 
 with the counting measures $\sum _i  \d_{(x_i,E_i)}$ and $\sum_i \d_{x_i}$, respectively.  
 On $\O$ one defines in a standard way a  metric   such that  the 
  $\s$--algebra $\cB(\O)$ of  Borel sets   is generated by the sets $\{ \o (A)= k\}$, with $A$ and  $k$ varying respectively among the Borel sets of $\bbR^d\times \bbR$  and the nonnegative integers \cite{DV}.

We consider  a 
 \emph{marked simple point process}, which  is a measurable function from a probability space to  the measurable space $\left(\O,\cB(\O)\right)$. 
 We denote by $\cP$  its law  and by $\bbE[\cdot]$ the associated expectation. $\cP$ is therefore a probability measure on $\O$.  We assume that $\cP$ is stationary and ergodic w.r.t. translations. 
 More precisely, given $x\in \bbR^d$ we define the translation $\t_x:\O \to \O$ as 
\[ \t_x \o:= \{ ( x_i -x, E_i)\}
\; \text{ if  }\;\o= \{ ( x_i , E_i)\}\,.
\]
 Then  stationarity means  that $\cP (\t_x A)=A$ for any $A\in \cB( \O)$, while ergodicity means that  $\cP(A)\in\{0,1\}$
 for any  $A\in \cB( \O)$ such that 
 $\t_x A=A$ for all $x\in \bbR^d$. 
 Due  to our  assumptions stated below, $\cP$ 
 has  finite positive  intensity $m$, i.e. 
  \begin{equation}\label{mom_palma0}
m:= \bbE\bigl[ \hat \o \bigl( [0,1]^d\bigr)\bigr]\in (0, +\infty)\,.
\end{equation}
   As a consequence,  the Palm distribution $\cP_0$ associated to $\cP$ is well defined \cite[Chp.~12]{DV}.  Roughly, $\cP_0$ can be thought as $\cP$ conditioned to the event $\O_0$, where    \be 
    \O_0:=\{ \o \in \O\,:\, 0 \in \hat \o\}\,.
   \en
  We will provide more details on $\cP$ and $\cP_0$ in  Section \ref{GS}.
Below, we write $\bbE_0[\cdot]$ for the expectation w.r.t. $\cP_0$.

\medskip
 
In addition to the marked simple point process with law $\cP$ we  fix  a nonnegative  Borel function
\[
\bbR^d\times \bbR^d \times \O \ni (x,y,\o) \mapsto c_{x,y} (\o)\in [0,+\infty)
\]
such that $c_{x,x}(\o)=0 $ for all $x\in \bbR^d$.
The value of $ c_{x,y} (\o)$ will be relevant only when $x,y \in \hat \o$.
For later use we define
    \begin{equation}\label{organic}
  \l_k(\o):=\int _{\bbR^d} d\hat \o (x) c_{0,x}(\o)|x|^k\,,
 \end{equation}
 where   $|x|$ denotes the  norm of $x\in \bbR^d$.

\begin{Definition}\label{fasma98}  We define the effective diffusion matrix $D$ as 
the $d\times d$ nonnegative symmetric matrix such that
 \begin{equation}\label{def_D}
 a \cdot Da =\inf _{ f\in L^\infty(\cP_0) } \frac{1}{2}\int d\cP_0(\o)\int d\hat \o (x) c_{0,x}(\o) \left
 (a\cdot x - \nabla f (\o, x) 
\right)^2\,,
 \end{equation}
 where $\nabla f (\o, x) := f(\t_x \o) - f(\o)$.
\end{Definition}
Above, and in what follows, we will denote by $a\cdot b$ the scalar product of the vectors $a$ and $b$.

\medskip

{\bf Assumptions.} 
We make the following assumptions:
\begin{itemize}
\item[(A1)] the law $\cP$ of the marked simple point process is stationary and ergodic w.r.t. spatial translations;
 \item[(A2)] $\cP$ has finite positive intensity as  stated in \eqref{mom_palma0};
\item[(A3)] $\cP ( \o \in \O:  \t_x\o\not = \t_y \o   \; \forall x\not =y \text{ in }\hat \o )=1$;
\item[(A4)] the weights $c_{x,y}(\o)$ are symmetric and covariant, i.e. $c_{x,y}(\o)=c_{y,x}(\o)$  $\forall x,y\in \hat\o$   and 
$
 c_{x,y}(\o) = c_{x-a, y-a}( \t_a \o)$  $\forall x,y \in \hat \o$ and $\forall a \in \bbR^d$;
\item[(A5)]   $\l_0, \l_2 \in L^1(\cP_0)$; 
 \item[(A6)]  for some $\a\in (0,1)$ it holds 
 \begin{align}
& \bbE_0\bigl[ \int d\hat \o (z) c_{0,z} (\o)^{\a}\bigr ] <+\infty\,, \label{sorpresina}\\
& \bbE_0\bigl[ \int d\hat \o (z) c_{0,z} (\o)^{\a}|z|^2 \bigr ] <+\infty\,, \label{zarina}\\
& \limsup  _{\ell  \to +\infty} \ell ^2 \sup_{\o \in \O_0}\sup_{z \in \hat \o: |z| \geq \ell} c_{0,z}(\o) ^{1-\a}   <+\infty 
\,; \label{downtown}
\end{align}
\item[(A7)]  $c_{x,y}(\o) >0$ for all $x,y\in \hat \o$.
\end{itemize}

We discuss the above assumptions at the end of this section.
\begin{Warning}\label{stellina59}
Since $D$ is a symmetric matrix, at cost of an orthonormal change of coordinates and  without loss of generality,
 we will  suppose that $D$ is diagonal.  In other words, our results refer to the principal directions of $D$.  Note that $a\in \bbR^d\setminus \{0\}$ is eigenvector of eigenvalue zero if $a\cdot D a =0$.
 \end{Warning}

In the rest, $\ell$ will be a positive number.
We consider the stripe   $S_\ell:=\bbR\times (-\ell/2,\ell /2)^{d-1}$ and the   box $\L_\ell:=(-\ell/2,\ell/2)^d$. We consider the  $\ell$--parametrized resistor network ${\rm (RN)}^\o_\ell$  on $S_\ell$  with  electrical filaments  defined as follows. To each unordered pair  $\{x,y\}$, such that  $x \in  \hat \o \cap \L_\ell$ and $y \in \hat \o \cap S_\ell$, 
we associate an electrical filament of conductivity $c_{x,y}(\o)$. We can think  of${\rm (RN)}^\o_\ell$ as a weighted unoriented graph  
with vertex set $ \hat \o \cap S_\ell$, edge set 
\be\label{latticini}
\bbB^\o_\ell:=\bigl\{\{x,y\}\,:\, x \in  \hat \o \cap \L_\ell\,,\; y \in \hat \o \cap S_\ell\,,\; x\not =y
\bigr\}
\en  and weight  of the edge $\{x,y\}$ given by the conductivity $c_{x,y}(\o)$, see Figure \ref{messicano1}.

\begin{figure}
\includegraphics[scale=0.5]{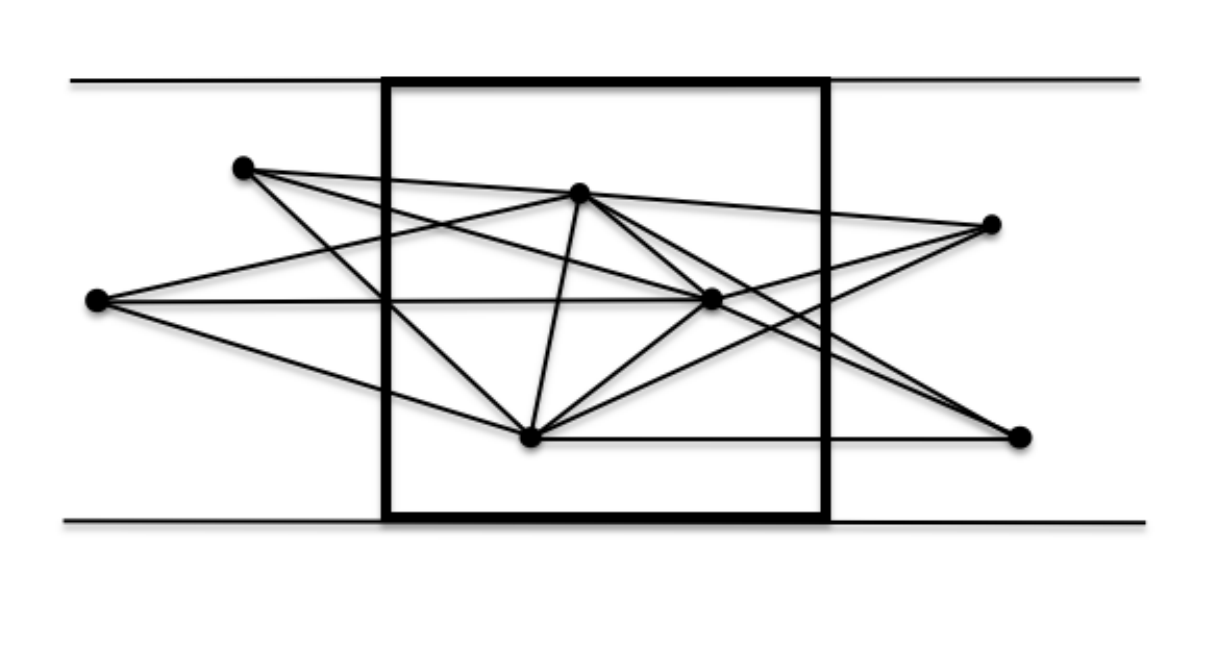}
\caption{A portion of the resistor network ${\rm (RN)}^\o_\ell$. The box and the stripe correspond to  $\L_\ell$ and   $S_\ell$, respectively. }\label{messicano1}
\end{figure}

Since the marked  simple point process is stationary and ergodic with positive intensity and $\bbE_0[\l_0]<+\infty$,  it is simple to prove that  there exists a translation invariant  Borel set $\O'\subset \O$  with   $\cP(\O')=1$ such that, for all  $\o\in \O'$ and  for all  $\ell \geq \ell_0(\o)$, it holds 
\begin{equation}
\label{pienezza}
\begin{split}
&\hat \o \cap \L_\ell \not = \emptyset \,,
\\
& \{ x\in  \hat \o\cap S_\ell \,:\, x_1\leq -\ell/2\} \not =\emptyset  \, ,\\
 & \{ x\in  \hat \o \cap S_\ell  \,:\, x_1\geq  \ell/2\} \not= \emptyset  \,,\\
 & {\sum} _{y\in \hat \o \cap S_\ell} c_{x,y}(\o) <+\infty \;\;\; \forall x \in \hat \o \cap\L_\ell\,.
 \end{split}
  \end{equation}
Indeed,  it is enough to apply Proposition~\ref{prop_ergodico} in Section \ref{GS} with suitable test functions $\varphi$, to bound the series in \eqref{pienezza} by ${\sum} _{y\in \hat \o} c_{x,y}(\o)  =\l_0(\t_x \o)$   and use that $\bbE_0[\l_0]<+\infty$.
\bigskip
\begin{Definition}[Electrical potential]  Suppose that $\o, \ell$ satisfy \eqref{pienezza}.  Then we denote by   $V_\ell^\o$  the \emph{electrical potential} of the resistor network (RN)$^\o_\ell$ with values $0$ and $1$ on $ \{ x\in \hat \o\cap S_\ell  \,:\, x_1\leq -\ell/2\}    $ and $\{ x\in \hat \o \cap S_\ell\,:\, x_1\geq  \ell/2\} $, respectively. In particular, $V_\ell^\o$ is the unique function $V_\ell^\o: \hat \o \cap S_\ell\to \bbR$ such that 
\be\label{thenero}
{\sum} _{y\in \hat \o \cap S_\ell} c_{x,y}(\o) \left(  V^\o _\ell (y)-  V^\o _\ell (x)\right) =0 \qquad \forall 
x\in \hat \o \cap \L_\ell\,,
\en
and satisfying the boundary conditions
\be\label{thebianco}
\begin{cases}
 V^\o_\ell (x) = 0  & \text{ if } x \in \hat \o \cap S_\ell, \;x_1\leq -\ell/2\,,\\
 V^\o_\ell (x) = 1  & \text{ if } x \in \hat \o \cap S_\ell, \;x_1\geq  +\ell/2\,.
\end{cases}
\en
\end{Definition}
As discussed in Section \ref{sec_hilbert}, the above  electrical  potential exists and is unique (here we use (A7)) and has values in $[0,1]$.
We recall that, given $(x,y)$ with $\{x,y\}\in \bbB^\o_\ell$ (cf. \eqref{latticini}),  
\be\label{ahahah} i_{x,y}(\o):=c_{x,y}(\o) \bigl( V^\o_\ell(y)- V^\o _\ell (x) \bigr)
\en is the current flowing from $x$ to $y$ under the electrical potential $V^\o_\ell$. For simplicity we have dropped the dependence on $\ell$ in the notation $i_{x,y}(\o)$.
\bigskip

\begin{Definition}[Effective conductivity] Suppose that $\o, \ell$ satisfy \eqref{pienezza}. 
We call $\s_\ell(\o)$ the \emph{effective conductivity} of the resistor network  $({\rm RN})^\o_\ell$ along the first direction under the electrical potential $V^\o_\ell$.  More precisely,  $\s_\ell (\o)$ is given by
 \be\label{ide1}
\begin{split}
 \s_\ell(\o)&: =  \sum _{\substack{x\in \hat \o \cap S_\ell:\\
 x_1\leq -\ell/2}} \; \sum_{y \in  \hat \o \cap \L_\ell} i_{x,y}(\o) =  \sum _{\substack{x\in \hat \o \cap S_\ell:\\
 x_1\leq -\ell/2}}\; \sum_{y \in  \hat \o \cap \L_\ell} c_{x,y} (\o) \bigl(V_\ell^\o (y)-V_\ell^\o(x)\bigr)\,.
\end{split}
 \en
\end{Definition} 

We recall two equivalent characterizations of the conductivity $\s_\ell(\o)$ (cf. Appendix \ref{ailo}).  For any $\g \in [-\ell/2, \ell/2)$, 
 $\s_\ell(\o)$ equals the current flowing through the hyperplane $\{x\in \bbR^d\,:\, x_1=\g\}$:
\begin{equation}\label{eq_ailo0}
\s_\ell(\o)=\sum _{\substack{x\in \hat \o \cap S_\ell:\\
 x_1\leq \g }} \; \sum_{\substack{ y \in  \hat \o \cap S_\ell:\\
 \{x,y\} \in \bbB^\o_\ell, \, y_1>\g
 }} i_{x,y}(\o) \,.
 \end{equation}
 Note that \eqref{ide1} corresponds to \eqref{eq_ailo0}  with $\g=-\ell/2$.
$\s_\ell(\o)$ also satisfies the identity
\begin{equation}\label{eq_ailo}
\s_\ell(\o)= \sum_{ \{x,y\} \in \bbB^\o_\ell } c_{x,y}(\o) \bigl( V^\o _\ell (x) - V^\o _\ell (y) \bigr)^2\,.
\end{equation}

We can now state our first main result concerning the infinite volume asymptotics of $\s_\ell (\o)$:
\begin{TheoremA}\label{teo1} For $\cP$--a.a. $\o$ it holds 
\be
\lim _{\ell \to +\infty} \ell^{2-d} \s_\ell(\o)= m D_{1,1}\,.
\en
\end{TheoremA}

To clarify the link with homogenization and state our further results, it is convenient 
to rescale space in order to deal with fixed stripe and box. More precisely, we set $\e:= 1/\ell$.
Then $\e>0$ is our scaling parameter. We set
\be\label{serpente}
\begin{cases}
S:=\bbR\times (-1/2,1/2)^{d-1}\,, &   \L:=(-1/2,1/2)^d\,,\\
S_-:= \{x\in S: x_1 \leq -1/2\} \,, &S_+:=\{x\in S: x_1 \geq 1/2\}\,.
\end{cases}
\en
We write  $V _\e: \e \hat \o \cap S \to [0,1] $ for the function given by $V_\e (\e x):= V_\ell ^\o (x)$ (note that the dependence on $\o$ in $V_\e$ is understood, as for other objects below). 

We introduce  the atomic measures 
\be\label{atomiche}
\mu^\e_{\o,\L}:= \e^d \sum _{x \in \e\hat \o  \cap \L} \d_x\,,\qquad 
 \nu ^\e_{\o,\L } : = 
  \sum_{(x,y)\in \cE_\e} \e^d  c_{x/\e,y/\e }(\o) \d_{( x ,(y-x)/\e)}\,,
\en
where   $\cE_\e$ is  the set  of  pairs  $(x,y)$ such that $x\not = y$ are in $\e \hat \o \cap S$ and $\{x,y\}$ intersect $\L$. Equivalently,  $\cE_\e:=\{ (\e x, \e y )\,:\, \{x,y\}\in \bbB_\ell^\o\}$.

Given a function $f: \e \hat \o \cap S \to \bbR$, we define the \emph{amorphous gradient} $\nabla_\e f$ on pairs $(x,z)$ with  $x\in \e \hat \o\cap S$ 
and $x+ \e z\in \e \hat\o\cap S$ as 
\begin{equation}\label{ricola}
 \nabla_\e f(x,z)= \frac{ f(x+\e z)- f(x)  }{\e}\,.
 \end{equation}
Moreover,  we define the operator
\be\label{degregori} \bbL^\e_\o f( x):= \e^{-2} \sum _{y  \in \e\hat \o \cap S }  c_{x/\e,y/\e}\left[ f( y) - f( x)\right]\,, \qquad x\in \e\hat \o  \cap \L\,,
\en
whenever  the series in the r.h.s. is  absolutely convergent. 

 Since $\bbE_0[\l_0]<+\infty$, we have $\cP_0(\l_0<\infty)=1$. By Lemma \ref{matteo} in Section \ref{GS} it  follows that $\cP(\O_1)=1$, where $\O_1$ is the 
translation invariant Borel set 
\be\label{bianchetto}
\O_1:=\{\o \in \O\,:\, \l_0 (\t_x \o) <+\infty \;\forall x\in \hat\o\}\cap \O'
\en
(see \eqref{pienezza} for the definition of $\O'$).
Let $\o \in \O_1$ and let  $f: \e \hat \o \cap S\to \bbR$  be a bounded function. 
Since $ \l_0 (\t_x \o)=\sum_{y\in \hat \o} c_{x,y}(\o)$, 
  $\bbL^\e_\o f(x) $ is well defined for all $x\in 
  \e \hat \o\cap \L$ and   the measure  $\nu^\e_{\o,\L}$ has finite mass ($\mu^\e_{\o,\L}$ has always finite mass as $\hat \o$ is locally finite).  As the amorphous gradient $\nabla_\e f$ is bounded too, we have that $\nabla_\e f\in 
 L^2(\nu^\e_{\o,\L})$. Moreover,  if in addition $f$ is zero outside $\L$,  it holds (cf. Lemma \ref{spiaggia})
\be\label{dir_form}
\la f, -\bbL^\e_\o f \ra _{L^2(\mu ^\e_{\o,\L})}
=
\frac{1}{2}\la \nabla _\e f, \nabla_\e f\ra _{L^2(\nu^\e_{\o,\L})}<+\infty \,.
\en
\begin{Definition}\label{def_hilbert} Given $\o \in \O_1$
we define  the Hilbert space 
\be
H^{1,\e}_{0,\o} :=\left\{ f :\e \hat \o \cap S \to \bbR \text{ s.t. } f(x)=0 \; \forall x  \in  \e \hat \o \cap (S_-\cup S_+) \right\}
\en endowed 
with norm
$\|f\|_{H^{1,\e}_{0,\o}}= \|f \| _{ L^2( \mu ^\e_{\o,\L})}+\| \nabla_\e  f \|_{L^2(  \nu ^\e_{\o,\L})}
$. In addition, we set $K^{\e}_{\o} :=H^{1,\e}_{0,\o}+ \psi $, where 
$\psi:S\to [0,1]$ is the function
\be\label{def_psi}
\psi(x):=
\begin{cases}
x_1+\frac{1}{2} & \text{ if } x\in \L\,,\\
0 & \text{ if }x\in S_-\,,\\
1 & \text{ if } x\in S_+\,.
\end{cases}
\en
\end{Definition}
Note that  $K^{\e}_\o$ is given by the functions  $ f:\e \hat \o \cap S \to \bbR$ such that 
 $f(x)=0$  for all $ x  \in  \e \hat \o \cap S_- $ and $f(x)=1$ 
 for all $ x  \in  \e \hat \o \cap S_+ $.  

Given    $\o \in \O_1$, in  Section \ref{sec_hilbert} we will derive that,  due to \eqref{thenero} and \eqref{thebianco},   $V_\e$ is the  unique function in $ K^\e _\o$ such that $\bbL^\e_\o V_\e (x)=0$ for all $x \in \e\hat \o \cap \L$ (cf. Lemma \ref{benposto}). We point out that, by \eqref{eq_ailo} and \eqref{dir_form},  the rescaled conductivity $\ell^{2-d} \s_\ell(\o)$  equals the flow energy associated to $V_\e$:
\be\label{chiave}
\ell^{2-d} \s_\ell(\o)= \la V_\e, -\bbL^\e_\o V_\e \ra_{L^2(\mu^\e_{\o,\L} )}=
 \frac{1}{2} \la   \nabla _\e V_\e, \nabla_\e V_\e \ra _{L^2(\nu ^\e_{\o,\L})} \,.
\en
Theorem \ref{teo1} can therefore  be restated as
\begin{equation}\label{mortisia} 
 \lim_{\e \da 0} \frac{1}{2} \la   \nabla _\e V_\e, \nabla_\e V_\e \ra _{L^2(\nu ^\e_{\o,\L})}= m D_{1,1}\la \nabla \psi, \nabla \psi\ra_{L^2(\L,dx)}=m D_{1,1}\,,\;\;\; \cP\text{-a.s.\,.}
 \end{equation}
 Note that the second identity in \eqref{mortisia} is immediate as $\nabla \psi = e_1$.
To prove Theorem \ref{teo1} we distinguish the cases $D_{1,1}=0$ and  $D_{1,1}>0$. The proof for $D_{1,1}=0$  (which is simpler) is given in Section \ref{renato_zero}, while  the proof for $D_{1,1}>0$ will take the rest of our investigation and will be concluded in Section \ref{limitone}.
In the case $D_{1,1}>0$ we can say more on the behavior of $V_\e$:
\begin{TheoremA}\label{teo2}
Suppose that $D_{1,1}>0$. Then 
there exists a translation invariant Borel set ${\O}_{\rm typ}$ of typical environments  with $ 
{\O}_{\rm typ} \subset  \O_1 $ and $\cP({\O}_{\rm typ})=1$, such that for any $\o \in {\O}_{\rm typ}$ 
\eqref{mortisia} holds, 
 $V_\e\in L^2(\mu^\e_{\o,\L}) $  converges  weakly and 2-scale converges weakly    to $\psi \in L^2(\L, dx)$.\end{TheoremA}
%
%
%
%
The definition of the above types of convergence is  recalled in 
 Section \ref{anatre12}.

\begin{Warning}\label{stellina5} Recall that $D$ is diagonal (see Warning \ref{stellina59}).
When $D_{1,1}>0$,   at cost to permute the coordinates and without loss of generality, we assume  that  $D_{i,i}>0$ for $1\leq i \leq d_*  $ and $D_{i,i}=0$ for $d_*<i \leq d$. 
 \end{Warning}

In  Section \ref{eff_equation} we will characterize $\psi$ as the unique weak solution on $\L$ of the so--called effective equation given by  $\nabla_* \cdot (D \nabla_* v)=0$  with suitable mixed Dirichlet-Neumann conditions, where $\nabla_*$ denotes the projection of $\nabla $ on the first $d_*$ coordinate  (cf. Definition \ref{fete}). Due to Theorem \ref{teo2}, the  equation $\nabla_* \cdot (D \nabla_* v)=0$ represents the effective macroscopic law of the electrical potential $V_\e$  in the limit $\e\da 0$, when $D_{1,1}>0$.

\subsection{Comments on Assumptions (A1),...,(A7)}

If the marked simple point process  is  the  $\nu$--randomization of an ergodic stationary simple point process $\xi $ on $\bbR^d$ (i.e.~under  $\cP( \cdot\,| \hat \o)$ the marks are i.i.d. with common law $\nu$) and $\nu$ is not degenerate  (i.e. $\nu \not = \d_a$), then condition (A1) is automatically satisfied (see \cite[Section~2.1]{FSS}).   The point process $\xi$ can be   genuinely amorphous  as 
the  Poisson point process or can keep some lattice structure 
 as the random set $\xi:=U+\tilde \xi\subset \bbR^d$, where $U$ and $\tilde \xi$ are independent, $U$  is a random vector with uniform distribution on $[0,1]^d$ and $\tilde \xi$ is given by the vertex set of a site/bond Bernoulli percolation in $\bbZ^d$.

 Always in the case of   $\nu$--randomization,
if  $\nu$ is not degenerate, then (A3) is also fulfilled.
In the general case, since the event in (A3) is translation invariant, (A3) is equivalent to the identity $\cP_0 ( \o \in \O_0: \t_x \o \not = \t_y \o \; \forall x\not=y \text{ in }\hat \o)=1$ (cf. e.g. \cite{DV}, \cite[Lemma 1]{FSS}).

 To verify (A5) and  
\eqref{sorpresina}, \eqref{zarina} in (A6) the following property is very useful: given $n\in \bbN$, $x\in \bbR^d$  and  a box $B\subset \bbR^d$,  it holds $\bbE_0[ \hat \o (x+B)^n] \leq  C \bbE[ \hat \o ([0,1]^d ) ^{n+1}] $
for some positive constant $C$ independent from $x$, cf.
 \cite[Lemma~1-(iv)]{FSS}. If, as  in Mott v.r.h., there exist $C'>0$ such that  $c_{0,x} (\o ) \leq C' f(|k|) $ for any $k\in \bbZ^d$ and  $x\in k+[0,1]^d$, then one can bound 
\[  \int \hat \o (z) c_{0,z}(\o) ^\g |z| ^{\chi} \leq 
C(\g,\chi)
\sum _{k \in\bbZ^d}f(|k|)^\g (1+|k|^\chi)  \hat \o (  k+[0,1]^d)\,.
\]
As a consequence,  if $\bbE[ \hat \o([0,1]^d) ^2]<+\infty$,  we have  $ \bbE_0\bigl[ \int d\hat \o (z) c_{0,z} (\o)^{\g}|z|^\chi \bigr ]<+\infty$   if $\sum _{k \in\bbZ^d}f(|k|)^\g (1+|k|^\chi) <+\infty$.
By Campbell's formula (take $f(x,\o) := \mathds{1}( \|x\| _\infty \leq 1/2) \hat \o([-1,1]^d)$ in  \eqref{campanello} below),  $\bbE_0[ \hat\o ( [-1,1]^d)]<+\infty$ implies that  $\bbE[ \hat\o ( [0,1]^d)^2 ]<+\infty$. In particular, for Mott v.r.h.  Assumption 
(A5), \eqref{sorpresina} and \eqref{zarina} are satisfied if and only if $\bbE[ \hat\o ( [0,1]^d)^2 ]<+\infty$.

Condition \eqref{downtown} can be relaxed. 
 For the sake of simplicity, and since \eqref{downtown} is true for Mott v.r.h., we have  preferred the present form.
Condition (A7) is not strictly necessary. It guarantees the  uniqueness of the electrical potential and it is always satisfied by  Mott v.r.h.\,.
Due to the above discussion, for Mott v.r.h., our assumptions  reduce to Assumptions (A1), (A2), (A3) and the requirement that   $\bbE[ \hat\o ( [0,1]^d)^2 ]<+\infty$.

Finally, we point out that the marks $E_i$  could indeed belong to any Polish space instead of $\bbR$, results and proofs would not change.
%
%


\section{Effective equation with mixed  boundary conditions}\label{eff_equation}
In this section we assume  that $D_{1,1}>0$.
  Recall the definition of $d_*$ given in Warning \ref{stellina5}. 
We are interested in elliptic operators with mixed (Dirichlet and  Neumann) boundary conditions.
We set 
\[  F_-:=\{ x\in \bar \L\,:\, x_1 =-1/2\}\,,\;  \,  F_+:=\{ x\in \bar \L\,:\, x_1 =1/2\}\,, \,
\; F:= F_- \cup F_+
\,.
\]
Given a domain  $A\subset \bbR^d$, $L^2(A)$ and $H^1(A)$ refer to the Lebesgue measure $dx$.

\begin{Definition}\label{vettorino} We introduce the following three functional spaces:
\begin{itemize}
\item We define $H^1(\L, d_*)$ as the  Hilbert space given by  functions $f\in L^2(\L)$ with weak derivative 
$\partial _i f$ in $L^2(\L)$ for any $i=1,\dots, d_*$,  endowed    with the 
 norm $\|f\| _{1,*}:= \|f\| _{L^2( \L)} +\sum _{i=1}^{d_*}   \|\partial _i f\| _{L^2( \L)}$. Moreover, given $f\in H^1(\L,d_*)$, we define
 \be\label{aria7}
 \nabla_* f:= (\partial _1 f, \partial _2 f, \dots, \partial _{d_*} f,0,\dots, 0) \,.
 \en
\item We define $H^1_0( \L,F,d_*)$ as the closure in  $H^1(\L,d_*)$ of
\[ 
\Big \{ \varphi _{|\L} \,:\, \varphi \in C^\infty_c (\bbR^d \setminus F) \Big\}\,.\]
 \item  We define the functional set $K$ as (cf. 
  \eqref{def_psi})
  \be\label{kafka}
K:= \{\psi_{|\L}+ f\,:\, f \in H^1_0( \L,F,d_*)\}\,.
\en
\end{itemize}
  \end{Definition}
  \begin{Remark}\label{memorandum} Let $f\in H^1(\L,d_*)$. Given $1\leq i \leq d_*$,  by integrating $\partial_i f $ times $\varphi (x_1, \dots, x_{d_*}) \phi (x_{d*+1}, \dots, x_d)$ with 
  $\varphi \in C_c^\infty (\bbR^{d_*})$ and $\phi \in C_c^\infty( \bbR^{d-d_*})$, one obtains  that  the function $f( \cdot, y_1, \dots, y_{n-d_*})$ belongs to $H^1\left( (-1/2,1/2)^{d_*}\right)$ for a.e. $(y_1, \dots, y_{n-d_*}) \in (-1/2,1/2)^{n-d_*}$.
  \end{Remark}
   Being a closed subspace of the Hilbert space
 $H^1(\L,d_*)$,  $H^1_0( \L ,F,d_*)$ is a Hilbert space.
We  also point out that in the definition of $K$ one could replace $\psi_{|\L}$ by any other function $\phi\in H^1(\L,d_*) \cap C(  \bar\L)$  such that $\phi\equiv 0$ on $F_-$ and $\phi\equiv 1$ on $F_+$, as follows from the next lemma: 
\begin{Lemma}\label{12anni} Let $u \in H^1(\L,d_*) \cap C(  \bar \L )$   satisfy $u\equiv 0$ on $F$. Then $u\in H^1_0( \L ,F,d_*)$.
\end{Lemma}
\begin{proof} We use some idea from the proof of \cite[Theorem~9.17]{Br}.
We set  $u_n(x):= G( n u(x) )/n$, where $G\in C^1(\bbR)$ satisfies: $|G(t) | \leq |t| $ for all $t\geq 0$, $G(t)=0$ for $|t|\leq 1$ and $G(t)=t $ for $|t|\geq 2$.  Note that   $\partial _i u_n(x) = G' (n u(x) ) \partial _i u(x) $ for $1\leq i \leq d_*$ (cf. \cite[Prop.~9.5]{Br}). Hence, $u_n\to u$ and $\partial _i u_n \to  \mathds{1}_{\{u=0\}} \partial_iu=\partial _i u$ a.e. In the last identity, we have used that $\partial _i u=0$ a.e. on $\{u=0\}$ which follows as a byproduct of  Remark \ref{memorandum} and  Stampacchia's theorem (see Thereom~3 and Remark~(ii) to Theorem 4 in \cite[Section~6.1.3]{EG}).  By dominated convergence one obtains that $u_n \to u$ in $H^1( \L, d_*)$. Since $H^1_0( \L, F,d_*)$ is a closed subspace of $H^1(\L,d_*)$, it is enough to prove that $u_n\in H^1_0( \L, F,d_*)$. Due to our hypothesis on $u$ and the definition of $G$, $u_n\equiv 0$ in a neighborhood of $F$ inside $ \bar \L$. Hence the thesis follows by applying the implication (iii)$\Rightarrow$ (i) in Proposition \ref{monti}. Equivalently, it is enough to observe that, by adapting  \cite[Cor.~9.8]{Br} or \cite[Theorem~1, Sec.~4.4]{EG}, there exists a sequence of functions   $\varphi_k \in C_c^\infty (\bbR^d)$ such that ${\varphi_k }_{|\L} \to u_n$ in $H^1(\L, d_*)$. Since $u_n\equiv 0$ in a neighborhood of $F$, it is easy to  find   $\phi \in C_c^\infty (\bbR^d\setminus F)$   such that  ${(\phi \varphi_k)} _{|\L} \to u_n$ in $H^1(\L, d_*)$. Hence  $u_n\in H^1_0( \L, F,d_*)$. 
\end{proof}

One can  adapt the proof of  \cite[Prop.~9.18]{Br} to get the following criterion assuring that a function belongs to $H^1_0(\L,F,d_*)$:
\begin{Proposition}\label{monti} 
Given  a function $u \in L^2(\L)$, the following properties are equivalent:
\begin{itemize}
\item[(i)] $ u \in H^1_0(\L, F,d_*)$;
\item[(ii)]
 there exists $C>0$  such that 
\be\label{tennis73}
\Big|
 \int _{\L} u\, \partial_{i}  \varphi dx
 \Big| 
 \leq C\|\varphi \|_{L^2(\L)} \quad \forall \varphi \in C_c^\infty (S) \,,\;\forall i:1 \leq i \leq  d_*\,;
 \en 
 \item[(iii)] the function 
 \begin{equation}
 \bar u (x) := 
 \begin{cases}
 u(x) & \text{ if } x \in \L\,,\\
 0 & \text{ if }  x \in S \setminus \L\,,
 \end{cases}
 \end{equation}
 belongs to $H^1(S,d_*)$ (which is defined similarly  to $H^1(\L, d_*)$. Moreover, in this case it holds $\partial_{i}\bar u = \overline{ \partial_{i} u}$ for $1\leq i \leq d_*$, where $ \overline{ \partial_{i} u}$ is defined similarly to $\bar{u}$.
 \end{itemize}
 \end{Proposition}

\begin{Lemma}[Poincar\'e inequality] \label{poincare}
It holds $\| f \|_{L^2(\L )}\leq \| \partial_1 f \|_{L^2(\L )}$ for any $f\in H^1_0 (\L, F,d_*)$.
\end{Lemma}
\begin{proof}
Given  $f \in \cC^\infty _c (\bbR^d\setminus F) $, by Schwarz inequality, for any $(x_1,x') \in \L$  we have 
$ f (x_1,x')^2=\bigl(
 \int_{-1/2} ^{x_1} \partial_1 f (s, x')ds\bigr)^2 \leq \int_{-1/2} ^{1/2} \partial_{1} f (s, x')^2 ds$. By integrating over $\L$ we get the desired estimate for   $f \in \cC^\infty _c (\bbR^d\setminus F) $.   Since $\cC^\infty _c (\bbR^d\setminus F )$ is dense in $H^1_0(\L, F,d_*)$, we get the thesis.
\end{proof}

\begin{Definition}\label{fete}
We say that  $v $ is a weak solution of the equation
\be\label{salvateci} 
\nabla_* \cdot ( D \nabla_* v ) =0
\en
on $\L$ with boundary conditions 
\be\label{mbc}
\begin{cases}
v(x) =0 & \text{ if } x\in F_-\,,\\
v(x)= 1 & \text{ if } x \in F_+\,,\\
 D \nabla_* v (x) \cdot  \mathbf{n}(x) =0 & \text{ if } x \in \partial \L \setminus F\,,
\end{cases}
\en
if $v \in K$ (cf. \eqref{kafka})  and if $\int_{\L} \nabla_* u \cdot D \nabla_* v\, dx =0$  for all  $u\in  H^1_0(\L, F,d_*)$.
\end{Definition}
Above $\mathbf{n}$ denotes the outward unit normal vector to  the boundary in $\partial \L$ (which is well defined on $\partial \L \setminus F$).
%
\begin{Remark}\label{piccoletto}
In the above definition it would be enough to 
require that $\int _\L \nabla_* u \cdot D \nabla_*  v dx =0$  for all   $u\in C^\infty_c(\bbR^d\setminus F) $ since the functional $H^1_0(\L , F,d_*)\ni u   \mapsto \int _\L\nabla_* u \cdot D \nabla_* v dx \in \bbR $ is continuous.
\end{Remark}
We shortly motivate the above definition. To simplify the notation we take $d_*=d$. We recall 
  Green's formula for a Lipschitz domain $B$: 
\be\label{lunare}
\int _B  ( \partial_i f) g\, dx=-\int_B f (\partial_i g)\, dx  + \int _{\partial B} fg  ( \mathbf{n}\cdot e_i ) d S\,, \qquad \forall f,g \in C^1(\bar{B})\,,
\en
where $\mathbf{n}$ denotes the outward unit normal vector to  the boundary $\partial B$ and $dS$ is the surface measure on $\partial B$.
 By taking $f= \partial _{j} v$ and $g=u$ in \eqref{lunare} we get 
\be \label{solare}
\int_B u \nabla \cdot ( D \nabla v ) dx  =-
\int_B \nabla u  \cdot ( D \nabla v) dx + \int _{\partial B} u (\nabla v \cdot (D\mathbf{n}) ) d S \,,
\en
for all $v \in C^2(\bar{B})$ and $ u \in C^1(\bar{B})$. By taking \eqref{solare} with $B=\L$ we see that $v\in C^2(\bar \L)$ satisfies 
$ \nabla \cdot ( D \nabla v )=0$ on $\L$ and  $\nabla v \cdot (D\mathbf{n})\equiv 0$ on $\partial \L \setminus F$ if and only if $\int_{\L} \nabla u  \cdot ( D \nabla v) dx=0 $ for any $u \in C^1(\bar \L)$ with $u\equiv 0 $ on $F$.
Such a set $\cC$ of  functions $u$ is  dense in $H^1_0(\L,F,d)$. Indeed $\cC\subset H^1_0(\L,F,d)$ by Lemma \ref{12anni}, while  $C^\infty_c(\bbR^d\setminus F)\subset \cC$. Hence, we conclude that $v\in C^2(\bar \L )$ satisfies 
$ \nabla \cdot ( D \nabla v )=0$  on $\L$ and  $\nabla v \cdot (D\mathbf{n})\equiv 0$ on $\partial \L \setminus F$ if and only if $\int_{\L} \nabla u  \cdot ( D \nabla v) dx=0 $ for any  $u\in  H^1_0(\L,F,d)$.  We have therefore proved that    $v\in C^2(\bar \L)$  is a classical solution of \eqref{salvateci} and \eqref{mbc} if and only if it is a weak solution in the sense of Definition \ref{fete}.
 
%
%
%
%
%
%
%
%
%
%
%
%
%


\begin{Lemma}\label{unico}
There exists a unique weak solution $u\in K$ of the equation $\nabla_* \cdot ( D \nabla_* u ) =0$ with boundary conditions  \eqref{mbc}. Furthermore, $u$ is the unique minimizer of 
\be
\inf _{v\in K}   \int \nabla_* v \cdot D \nabla_* v \,dx\,.
\en
\end{Lemma}
\begin{proof} To simplify the notation, in what follows we write $\psi$ instead of $\psi_{|\L}$.
We define the bilinear form $a(f,g):= \int _{\L} \nabla_* f \cdot D \nabla_* g dx $ on  the Hilbert space $H^1_0(\L,F,d_*)$. The bilinear form $a(\cdot,\cdot)$ is symmetric and  continuous (since $D$ is symmetric).  Due to the  Poincar\'e inequality (cf.~Lemma~\ref{poincare}) and since $D_{1,1}>0$, 
$a(\cdot,\cdot)$ is also coercive. 


By definition 
we have that $u\in K$ is a weak solution of  equation $\nabla_* \cdot ( D \nabla_* u ) =0$ with b.c. \eqref{mbc} if and only if, setting $f:=u-\psi$, $f \in H^1_0(\L, F,d_*)$ and $f$ satisfies 
\be\label{baia}
\int \nabla_* f\cdot D \nabla_* v dx = - \int \nabla_* \psi  \cdot D \nabla_* v\, dx \qquad \forall v \in H^1_0(\L, F,d_*)\,. 
\en
Note that the r.h.s. is a continuous functional in $v\in H^1_0(\L, F,d_*)$.
Due to the above observations and by  Lax--Milgram theorem  we conclude that there exists a unique such function  $f$, hence there is a unique weak solution $u$ of  equation $\nabla_* \cdot ( D \nabla_* u ) =0$ with b.c. \eqref{mbc}. Moreover $f$ satisfies  \be
\frac{1}{2} a(f,f)+ \int \nabla_* \psi \cdot D \nabla_* f \, dx= \inf _{g\in H^1_0(\L, F,d_*)} \Big\{\frac{1}{2} a(g,g)+ \int \nabla_* \psi  \cdot D \nabla_* g\,dx
\Big\}\,.
\en
By adding to both sides $\frac{1}{2} \int \nabla_* \psi \cdot D \nabla_* \psi  \, dx$,  we get that  $\frac{1}{2}\int \nabla_* u \cdot D\nabla_* u= \inf _{v\in K} \frac{1}{2}\int \nabla_* v \cdot D\nabla_* v \,dx$.
\end{proof}
From the above lemma we immediately get:
\begin{Corollary}\label{h2o} The function $\psi_{|\L}$ (cf.  \eqref{def_psi} ) is the unique weak solution 
  $u\in K$ of the equation $\nabla_* \cdot ( D \nabla_* u ) =0$ with boundary conditions  \eqref{mbc}.
\end{Corollary}

\section{Preliminary facts on $\O$, $\cP$ and $\cP_0$}\label{GS}

In this section we  recall some basic  facts on the space $\O$ and on the Palm distribution $\cP_0$ associated to $\cP$.

The space $\O$  of realizations of marked point processes is endowed 
with a Prohorov-like metric $d$ 
 such that the following facts are equivalent: 
 (i)   a sequence  $(\o_n)$ converges to  $\o$ in $(\O, d)$, (ii) 
$\lim _{n\to \infty} \int _{\bbR^d\times \bbR} f(x,s) d\o_n (x,s) = \int_{\bbR^d\times \bbR}  f (x,s) d \o (x,s) \,,
$ for any bounded continuous function $f: \bbR^d  \times \bbR\to \bbR$ vanishing outside a bounded set and (iii) 
$\lim_{n \to \infty} \o _n (A)=\o(A)$ for any bounded  Borel set  $A\subset \bbR^d \times \bbR $ with $\o(\partial A)=0$
 (see \cite[App.~A2.6 and Sect.~7.1]{DV}). 
  In addition, $(\O,d)$ is a separable metric space. Indeed, the above  distance $d$ is defined on the larger space $\cN$ of counting measures  $\mu=\sum _{i} k_i \d_{(x_i,E_i)}$,  where $k_i\in \bbN$  and $\{(x_i,E_i)\} $ is a locally finite subset of $\bbR^d\times \bbR$, and one can prove that $(\cN,d)$ is a Polish space having $\O$ as Borel subset \cite[Cor.~7.1.IV, App.~A2.6.I]{DV}.

 We recall some properties of the Palm distribution $\cP_0$ associated to the measure $\cP$ on $\O$. 
     $\cP_0$ is a probability measure with support inside  
    $\O_0$ and it can be  characterized by the identity 
\begin{equation}\label{marlena}
 \cP_0(A)= \frac{1}{m} \int _{\O}\cP(d\o)\int_{[0,1]^d}  d\hat{\o} (x) \mathds{1}_A(\t_x \o)\,, \qquad \forall A\subset \O_0 \text{ Borel}\,.
 \end{equation}
  The above identity  \eqref{marlena}  is a special case of the so--called 
 Campbell's formula (cf.~\cite[Eq.~(12.2.4)]{DV}): for any  nonnegative Borel function $f: \bbR^d\times \O \to[0,\infty) $ it holds (recall \eqref{mom_palma0})
 \begin{equation}\label{campanello}
 \int_{\bbR^d}dx  \int _{\O_0} \cP_0 ( d\o) f(x, \o) =\frac{1}{m} \int _{\O}\cP(d\o)\int_{\bbR^d}  d\hat{\o} (x) f(x, \t_x \o) \,.
 \end{equation}
  An alternative characterization of $\cP_0$ is  described in \cite[Section 1.2]{ZP}.

A fact frequently used in the rest  is the following (see \cite[Lemma 1]{FSS}): given a  translation invariant  Borel  subset $A\subset \O$,
it holds  $\cP(A)=1 $ if and only if $ \cP_0(A)=1$.

We recall some basic technical  facts discussed in \cite{Fhom}:

\begin{Lemma}\label{matteo}  
  \cite[Lemma 4.1]{Fhom}
Given a Borel subset $A\subset \O_0$, the following facts are equivalent:
\begin{itemize}
\item[(i)] $\cP_0(A)=1$;
\item[(ii)] $\cP\left(\o \in \O\,:\, \t_x \o \in A \; \forall x \in \hat\o\right)=1$;
\item[(iii)]  $\cP_0\left(\o \in \O_0\,:\, \t_x \o \in A \; \forall x \in \hat\o\right)=1$.
\end{itemize}
\end{Lemma}

\begin{Lemma}\label{lucertola1}\cite[Lemma 1--(i)]{FSS}\cite[Lemma~4.3]{Fhom} Let $k:\O_0\times \O_0 \to \bbR$ be a Borel function such that (i)  at least one of the functions   $\int d \hat{\o} (x) |k(\o, \t_x \o) | $ and $\int d\hat{\o}(x) |k(\t_x\o,\o)|$ is  in $L^1(\cP_0)$, or (ii) $k(\o,\o')\geq 0$. Then
\begin{equation}\label{prugna1}
\int d \cP_0(\o) \int d\hat{\o}(x) k(\o, \t_x \o) = \int d \cP_0(\o) \int d \hat{\o}(x) k (\t_x\o, \o)\,.
\end{equation}
\end{Lemma}

We conclude by  focusing on ergodicity.
  Since by Assumption (A1)  $\cP$ is ergodic,  we 
 have the following  result (cf. \cite[Prop.~12.2.VI]{DV}): given a nonnegative Borel function $g: \O_0\to [0,\infty)$  it holds
 \begin{equation}\label{lattone1}
 \lim_{n \to \infty} \frac{1}{(2n)^d } \int_{[-n,n]^d  }  d\hat{\o}(x) \, g(\t_x \o) = m \,\bbE_0[ g ]\qquad \cP\text{--a.s.}\,.
 \end{equation}
 One can indeed refine the above result. To this aim 
we define $\mu^\e_\o$ as the atomic measure on $\bbR^d$ given by  $\mu^\e_\o := \e^d \sum _{x\in \hat \o} \d_{\e x}$. Then it holds:
\begin{Proposition}\label{prop_ergodico}\cite[Prop.~3.1]{Fhom} Let  $g: \O_0\to \bbR$ be a Borel function with $\|g\|_{L^1(\cP_0)}<+\infty$. Then there exists a translation invariant   Borel subset $\cA[g]\subset \O$  such that $\cP(\cA[g])=1$ and such that,  for any $\o\in \cA[g]$ and any  $\varphi \in C_c (\bbR^d)$, it holds
\begin{equation}\label{eq-limitone}
\lim_{\e\da 0} \int  d  \mu_\o^\e  (x)  \varphi (x ) g(\t_{x/\e} \o )=
\int  dx\,m\varphi (x) \cdot \bbE_0[g]\,.
\end{equation}
\end{Proposition}
The above proposition (which is the analogous e.g. of \cite[Theorem~1.1]{ZP}) is at the core of 2-scale convergence. It corresponds to  a refined version of ergodicity. The variable $x$ appears in the l.h.s. of \eqref{eq-limitone} at  the macroscopic scale in $\varphi (x)$ and at the microscopic scale in $g(\t_{x/\e} \o )$.

\begin{Definition}\label{sibilla}  Given a function  $g: \O_0 \to [0,+\infty]$ such that  $\bbE_0[ g] <+\infty$,  we define $\cA[g]$ as $\cA[g_*]$ (cf. Proposition \ref{prop_ergodico}), where  $g_*: \O_0 \to \bbR$ is defined as $g$ on $\{ g <+\infty\}$ and as $0$ on $\{ g =+\infty\}$. 
\end{Definition}


\section{The Hilbert space $H^{1,\e}_{0,\o}$ and  the amorphous gradient $\nabla_\e f$} \label{sec_hilbert}
In this section we come back to the Hilbert space $H^{1,\e}_{0,\o}$ introduced in Section \ref{MM}, proving  some properties used there and extending the discussion.
In addition, in Subsection \ref{gnocchi} we collect  some basic   properties of  the amorphous gradient $\nabla_\e$, which will be frequently used in the proof of Theorem \ref{teo2}.

Let $\o \in \O_1$ (cf. \eqref{bianchetto}). Recall Definition \ref{def_hilbert} of $H^{1,\e} _{0,\o} $  and $K^\e_\o$.
As discussed in Section \ref{MM}, if $f:\e \hat \o \cap S\to \bbR$ is bounded, then $f\in L^2(\mu^\e_{\o,\L} )$, $\nabla_\e f \in L^2(\nu^\e_{\o,\L})$ and $\bbL^\e _\o f\in   L^2(\mu^\e_{\o,\L})$. By definition of $\nu^\e_{\o,\L}$, given bounded functions $f,g:\e \hat \o\cap S \to \bbR$, we have 
\be\label{violino5}
\la \nabla_\e f, \nabla_\e g \ra _{L^2(\nu^\e_{\o,\L} )}= \e^{d-2} \sum _{ (x,y) \in \cE_\e  }c_{x/\e,y/\e}(\o) \bigl( f(y)-f(x) \bigr)  \bigl( g(y)-g(x) \bigr)\,.
\en

\begin{Lemma}\label{spiaggia} Let $\o \in \O_1$. Given $f,g: \e \hat \o\cap S \to\bbR$ with 
$f  \in H^{1,\e} _{0,\o} $ and  $g$ bounded, it holds
\be\label{gallina}
\la f, -\bbL^\e_\o g \ra_{L^2(\mu^\e_{\o,\L} )}=
 \frac{1}{2} \la   \nabla _\e f, \nabla_\e g  \ra _{L^2(\nu ^\e_{\o,\L})} \,.
 \en
\end{Lemma}
\begin{proof}
Since $f\equiv 0$ outside $\L$ 
we have 
\be\label{consegna}
\begin{split}
\la f, -\bbL^\e_\o g \ra_{L^2(\mu^\e_{\o,\L} )}=
- \sum _{x\in \e \hat \o \cap S} \e^{d-2} f(x) \sum _{y \in \e \hat \o \cap S} c_{x/\e,y/\e}(\o) \bigl( g(y)- g(x) \bigr)\,.
\end{split} 
\en
The r.h.s. is an absolutely convergent series  as $\o \in \O_1$, hence we can freely permute the addenda.
Due to the symmetry of the jump rates, the r.h.s. of \eqref{consegna} equals
\[-\sum _{y\in \e \hat \o \cap S} \e^{d-2} f(y) \sum _{x \in \e \hat \o \cap S} c_{x/\e,y/\e}(\o) \bigl( g(x)- g(y) \bigr)\,.
\] By summing  the above expression with  the r.h.s. of \eqref{consegna}, we get
\be\label{siriano}
\la f, -\bbL^\e_\o g \ra_{L^2(\mu^\e_{\o,\L} )}=\frac{1}{2}
  \e^{d-2} \sum _{x\in \e \hat \o \cap S }
  \sum_{y\in \e \hat \o \cap S} c_{x/\e,y/\e}(\o) \bigl( f(y)-f(x) \bigr)  \bigl( g(y)-g(x) \bigr)\,.
\en
As the generic addendum in the r.h.s. is zero if $(x,y)\not \in \cE_\e $ since $f\equiv 0$ on $S\setminus \L$,  by \eqref{violino5} we get \eqref{gallina}.
\end{proof}
%
%
%
%
%
 \begin{Warning}\label{gubana_piena}
 In the following lemma, and in the rest,  when 
 considering $\o \in \O_1$ we will restrict (without further mention) to $\e$ small enough to 
 satisfy \eqref{pienezza} with $\ell=\e^{-1}$.
  \end{Warning}
\begin{Lemma}\label{benposto} Given   $\o \in \O_1$, 
 the following holds:
\begin{itemize}
\item[(i)]
There is a unique function $V_\e \in K^\e _\o$ such that $\bbL^\e_\o V_\e (x)=0$ for all $x \in \e\hat \o \cap \L$.
\item[(ii)]  $V_\e$ is the unique  function $v\in K^\e_\o  $ such that 
$
\la  \nabla_\e u,\nabla_\e v\ra _{L^2(\nu ^\e_{\o,\L} ) }=0$ for all $u \in H^{1,\e}_{0,\o}$.
\item[(iii)] $V_\e$ is the unique minimizer of  the following variational problem:
\be\label{avengers}
\inf \Big\{  \la \nabla_\e v, \nabla _\e v\ra _{L^2(\nu^\e_{\o,\L})}\,\Big{|}\, v \in K^\e_\o\Big\}\,. 
\en
\end{itemize}
\end{Lemma}
\begin{proof} On the finite dimensional Hilbert space $ H^{1,\e}_{0,\o}$
we consider the bilinear form $a(f,g):= \frac{1}{2}\la \nabla _\e f, \nabla _\e g \ra_{L^2(\nu^\e_{\o,\L})}$. Trivially, $a(\cdot, \cdot)$ is a continuous symmetric form. Moreover, by Assumption (A7) and \eqref{pienezza}, it holds $a(f,f)=0$ 
if and only if $f\equiv 0$  (see Warning \ref{gubana_piena}).  As a consequence, 
the bilinear form $a(\cdot, \cdot)$ is also coercive. 
By writing
$V_\e= f_\e +\psi$, the function $V_\e$ in Item (i) is the only one such that   $f_\e \in H^{1,\e}_{0,\o}$ and 
\be\label{sibillino}
\bbL^\e_\o f_\e( x) = - \bbL^\e _\o \psi 
\qquad \forall   x \in \e \hat \o \cap \L\,.\en
Due to   Lemma \ref{spiaggia}   $f_\e\in H^{1,\e}_{0,\o}$ satisfying \eqref{sibillino}  can be characterized also as the solution  in  $H^{1,\e}_{0,\o}$  of the problem 
\be \label{marino}
a( f_\e, u)= -
 \frac{1}{2}\la \nabla _\e \psi , \nabla _\e u \ra_{L^2(\nu^\e_{\o,\L})}
 \qquad \forall u \in  H^{1,\e}_{0,\o}\,.
\en
By the Lax--Milgram theorem we conclude that there exists a unique function $f_\e$ satisfying \eqref{marino}, thus implying Item (i).  
Since $a( f_\e, u)= \frac{1}{2}\la \nabla_\e f_\e , \nabla_\e u \ra_{L^2(\nu^\e_{\o,\L})} 
$, the uniqueness of the solution $f_\e$ of \eqref{marino} corresponds to Item (ii).
Moreover, always by  the  Lax--Milgram theorem,  $f_\e$ is the unique minimizer of  the functional 
$H^{1,\e}_{0,\o}    \ni v \mapsto \frac{1}{2} a( v,v) + \frac{1}{2}\la \nabla _\e \psi , \nabla _\e v \ra_{L^2(\nu^\e_{\o,\L})}$, and therefore of the functional 
$H^{1,\e}_{0,\o}    \ni v \mapsto\frac{1}{4} \la  
  \nabla_\e (v+\psi) , \nabla_\e (v+\psi) \ra _{L^2(\nu^\e_{\o,\L})}$.
 This proves Item (iii).
 \end{proof}

\begin{Remark}\label{fiorellino3} As  $V_\e$ is  ``harmonic'' on $\e \hat \o \cap \L$ (cf. Lemma \ref{benposto}--(i)) and $\o \in \O_1$,  $V_\e$ has values in $[0,1]$.
\end{Remark}


\begin{Lemma}\label{paletta}  There exists a translation invariant Borel subset $\O_2\subset \O_1$ such that $\cP(\O_2)=1$ and,   for all $\o\in \O_2$,
\begin{align} 
& \limsup_{\e \da 0}   \| \psi\|_{L^2(  \mu^\e_{\o, \L})} <+\infty \,, \qquad\limsup_{\e \da 0}   \|\nabla_\e \psi\|_{L^2(\nu^\e_{\o,\L})} < +\infty\,,\label{parasole0}
\\
&  \limsup_{\e \da 0}   \| V_\e \|_{L^2(\mu^\e_{\o,\L})} <+\infty \,, \qquad\limsup_{\e \da 0}   \|\nabla_\e V_\e\|_{L^2(\nu^\e_{\o,\L})} < +\infty\,.\label{parasole}
\end{align}
\end{Lemma} 
\begin{proof}
By Proposition \ref{prop_ergodico} applied with suitable test functions $\varphi$, there exists a
 translation invariant Borel set $\O_2\subset \O_1$ such that 
 $\lim_{\e  \da 0}\mu^\e _\o (\L) = m $ and $\lim_{\e \da 0} \int_\L \mu^\e_\o (dx) \l_2 (\t _{x/\e}\o) = \bbE_0[\l_2] $ for any $\o \in \O_2$.

Let us take $\o \in \O_2$. Since  $\psi, V_\e$ have value in $[0,1]$ and $ \mu^\e_{\o,\L}$ has mass $\mu^\e_\o (\L)\to m$, we get the first bounds in \eqref{parasole0} and  \eqref{parasole}.

 Let us prove that $ \limsup_{\e \da 0}   \|\nabla_\e \psi\|_{L^2(\nu^\e_{\o,\L})} < +\infty$.
 We have  (recall \eqref{violino5})
\be\label{carretto}
\begin{split}
 \|  \nabla_\e \psi \|^2  _{L^2(\nu^\e_{\o,\L})}& =\e^{d-2} \sum _{ (x,y) \in \cE_\e } c_{x/\e,y/\e}(\o) \left( \psi( y) -\psi (x) \right)^2\\
 & \leq 
 \e^{d-2} \sum _{ (x,y)  \in \cE_\e } c_{x/\e,y/\e}(\o) \left( y_1-x_1\right)^2 \\
 &\leq 2 \e^{d-2} \sum _{x\in \e\hat \o  \cap \L} \sum _{ y \in \e \hat \o \cap S} c_{x/\e,y/\e}(\o)  \left(y_1 -x_1 \right)^2\,.
  \end{split}
\en
We can rewrite the last expression as \[2 \e^d  \sum _{x\in \hat \o  \cap (\e^{-1} \L) } \sum _{ y \in  \hat \o \cap(\e^{-1} S)} c_{x ,y}(\o)  (y_1-x_1)^2
 \,,\]
which is upper  bounded by $ 2\e^d  \sum _{x\in \hat \o  \cap (\e^{-1} \L) } \l_2(\t_x \o)=2\int_\L \mu^\e_\o (dx) \l_2(\t _{x/\e} \o) $. The last integral  converges to $2\bbE_0 [\l_2]<+\infty $ as $\o \in \O_2$. This concludes the proof that 
$ \limsup_{\e \da 0}   \|\nabla_\e \psi\|_{L^2(\nu^\e_{\o,\L})} < +\infty$.

 Since $V_\e $ minimizes \eqref{avengers}, 
 we have  $  \|\nabla_\e V_\e\|_{L^2(\nu^\e_{\o,\L})}  \leq   \|\nabla_\e \psi \|_{L^2(\nu^\e_{\o,\L})} $. Hence 
   $\limsup_{\e \da 0} \|\nabla_\e V_\e\|_{L^2(\nu^\e_{\o,\L}) }  <+\infty $ 
   by the second bound in  \eqref{parasole0}.
\end{proof}

\subsection{Some properties of the amorphous gradient $\nabla_\e$}\label{gnocchi} In Section \ref{MM} we have defined $\nabla _\e f $ for functions $f: \e \hat \o\cap S\to \bbR$. The definition  can by extended by replacing $S$ with any set $A\subset \bbR^d$.
Given $f,g : \e \hat \o \to \bbR$, it is simple to check the following Leibniz rule:
 \begin{equation}\label{leibniz} \nabla _\e  (fg)(x,z)
   =\nabla _\e  f (x, z ) g (x )+ f (x+\e z ) \nabla _\e g  ( x, z)\,.
\end{equation}
  Let $\varphi \in C^1_c(\bbR^d)$. 
 Let $\ell$ be such that $\varphi (x)=0$ if $|x| \geq \ell$. Fix   $\phi \in C_c(\bbR^d)$ with values in $[0,1]$, such that $ \phi(x)=1$ for $|x| \leq \ell$ and $\phi(x)=0$ for $|x| \geq \ell+1$. Since $\nabla _\e \varphi(x,z)=0$ if $|x| \geq \ell$ and $|x+\e z|\geq \ell$, by the    mean value theorem we conclude that 
\be\label{paradiso}
\bigl | \nabla _\e \varphi(x,z) \bigr | \leq \| \nabla \varphi \|_\infty |z| \bigl( \phi(x)+ \phi(x+\e z) \bigr) \,.
\en
If in addition   $\varphi \in C_c^2(\bbR^d)$,    by Taylor expansion  $|\nabla_\e \varphi (x,z) -\nabla \varphi(x) \cdot z| \leq  \e C(\varphi)  |z|^2$ for some constant $C(\varphi)$ depending only on $\varphi$.
Note  that 
$\nabla_\e \varphi (x,z) -\nabla \varphi(x) \cdot z=0$ if $|x| \geq \ell$ and $|x+\e z| \geq \ell$. Hence we get   that 
\be \label{mirra}
\bigl| \nabla _\e \varphi (x,z) - \nabla \varphi (x) \cdot z\bigr | \leq \e C(\varphi) |z|^2 \bigl( \phi(x) + \phi(x+\e z) \bigr)\,.
\en

%
%
%
%

\section{Proof of Theorem \ref{teo1} when  $D_{1,1}=0$}\label{renato_zero}
We need to prove \eqref{mortisia}, i.e. that   $\cP$--a.s. 
$\lim_{\e \da 0}  \la \nabla_\e V_\e, \nabla_\e V_\e\ra_{L^2(\nu^\e_{\o,\L})}=0$. As $D_{1,1}=0$ and by \eqref{def_D},
given $\d>0$ we can  fix $f\in L^\infty(\cP_0)$ such that 
\be\label{moonstar}
\bbE_0\Big[\int d\hat \o (x) c_{0,x}(\o) \left (x_1 - \nabla f (\o, x) 
\right)^2\Big]\leq \d\,.
\en
Given $\e>0$ we define the function $v_\e: \e \hat \o \cap S\to \bbR$ as
\be
v_\e(x):=
\begin{cases}
\psi(x) + \e f (\t_{x/\e}\o ) & \text{ if } x\in \L\,,\\
0 & \text{ if } x \in S_-\,,\\
1 & \text{ if } x\in S_+\,.
\end{cases} 
\en
By Lemma \ref{benposto}-(iii) it is enough to prove that  $\lim_{\e \da 0}  \la \nabla_\e v_\e, \nabla_\e v_\e\ra_{L^2(\nu^\e_{\o,\L})}=0$ $\cP$--a.s.\,.
We write 
\begin{equation}
  \frac{1}{2}\la \nabla_\e v_\e, \nabla_\e v_\e\ra_{L^2(\nu^\e_{\o,\L})}
 \leq \e^{d-2}
\sum _{x\in  \hat \o \cap  \e^{-1}\L}  \sum _{y \in  \hat \o \cap  \e^{-1} S} c_{x, y}(\o)
\bigl( v_\e (\e y) - v_\e(\e x) \bigr)^2\,.
\end{equation}
We split the sum in the r.h.s.  into three contributions $C(\e)$, $ C_-(\e)$ and $C_+(\e)$, corresponding respectively to the cases $ y \in \hat \o \cap \e^{-1} \L$,  $ y \in \hat \o \cap \e^{-1} S_-$ and
$ y \in \hat \o \cap \e^{-1} S_+$, while in all the above contributions $x$ varies among $\hat \o \cap  \e^{-1}\L$.

If $x,y\in \hat \o \cap \e^{-1} \L$, then $v_\e(y)-v_\e(x)= \e( y_1-x_1- \nabla f( \t_x \o, y-x)) $. Hence, we can bound
\be
C(\e) \leq \e^d \sum _{x\in  \hat \o \cap  \e^{-1}\L} \; \sum _{y \in  \hat \o } c_{x, y}(\o)( y_1-x_1- \nabla f( \t_x \o, y-x)) ^2 \,.
\en
By ergodicity  (cf. \eqref{lattone1}, Proposition \ref{prop_ergodico}) 
the r.h.s.  converges $\cP$--a.s.  to  the l.h.s of \eqref{moonstar}, and therefore it is bounded by $\d$ $\cP$--a.s. Hence, $\varlimsup _{\e \da 0} C(\e) \leq \d$.

We now consider $C_-(\e)$ and prove that $\lim_{\e \da 0} C_-(\e)=0$.
If $x\in  \hat \o \cap \e^{-1} \L$ and $y \in \hat \o \cap \e^{-1} S_-$, then $(v_\e(x)-v_\e(y))^2 = \e^2(  x_1 - f (\t_x \o))^2 \leq 2 \e^2 x_1^2 + 2 \e^2 \|f\|_\infty^2 \leq 2 \e^2 (x_1-y_1)^2 + 2 \e^2 \|f\|_\infty^2$. Hence it remains to show  that 
\be\label{nike}
\e^d \sum _{x\in  \hat \o \cap  \e^{-1}\L} \; \sum _{y \in  \hat \o \cap  \e^{-1} S_-} c_{x, y}(\o)[( x_1-y_1)^2+ 1]
\en
goes to zero as $\e\da 0$.   
Given $\rho\in (0,1/2)$ we set $\L_\rho:=(-\rho, \rho)^d$.  We denote by 
$A_1(\rho, \e)$    the sum in \eqref{nike} restricted to $x\in \hat \o \cap \e^{-1}\L_\rho $ and  $y \in  \hat \o \cap  \e^{-1} S_-$. We denote by  $A_2(\rho, \e)$ the sum coming from the remaining addenda  so that \eqref{nike} equals $A_1(\rho, \e)+A_2(\rho, \e)$.
Given  $x,y$ as in  $A_1(\rho, \e)$, it holds $x_1-y_1\geq (1/2-\rho)/\e\geq 1$ for $\e $ small enough. In this case,       we can bound
$ c_{x, y}(\o)[( x_1-y_1)^2+ 1]\leq  C c_{x,y}(\o)^\a $, for some   universal positive constant $C$. Indeed, due to \eqref{downtown}, $\lim _{\ell \to +\infty} \ell^2 \rho(\ell)<+\infty$ where 
$\rho (\ell) := \sup_{\o \in \O_0} \sup_{z\in \hat \o: |z|=\ell}  c_{0,z} (\o )^{1-\a}$.
Due to the above observations,
\be
A_1(\rho,\e)\leq   C(\o) \e^d \sum _{x\in  \hat \o \cap  \e^{-1}\L} \; \sum _{y \in  \hat \o } c_{x, y}(\o) ^\a \mathds{1} ( |x-y| \geq \rho /\e)\,.
\en
By the ergodic theorem and \eqref{sorpresina}, we  get that $\lim_{\e \da 0} A_1 (\rho, \e)=0$ $\cP$--a.s.
We move to $A_2(\rho, \e)$.  We can bound $A_2(\rho, \e)$ by 
\be\label{nike7}
\e^d \sum _{x\in  \hat \o \cap  \e^{-1}(\L\setminus \L_\rho)} \sum _{y \in  \hat \o } c_{x, y}(\o)[(x_1-y_1)^2+ 1]\,.
\en
By Proposition \ref{prop_ergodico} with suitable test functions, we get that \eqref{nike7} converges as $\e \da 0$ to $\bbE_0 [ \l_2+ \l_0]  \ell (\L\setminus \L_\rho))$, where here $\ell(\cdot)$ denotes the Lebesgue measure. To conclude the proof that $\lim_{\e\da 0} C_-(\e)=0$,  it is therefore enough to take the limit $\rho \uparrow 1/2$.

By the same arguments used for $C_-(\e)$, one proves that $\lim_{\e\da 0} C_+(\e)=0$.


\section{Square integrable forms and effective diffusion matrix}\label{sec_review}
\begin{Warning} From this section, until  Section \ref{limitone} included, we assume that $D_{1,1}>0$. In particular, $d_*\geq 1$ is defined according to Warning \ref{stellina5}.
\end{Warning}

As typical  in homogenization theory \cite{JKO}, the variational formula \eqref{def_D} defining the effective diffusion matrix $D$  admits 
a geometrical interpretation  in the Hilbert space of square integrable forms.
We  recall here this interpretation. We also   collect  some  facts taken from \cite{Fhom}. They are  mainly  an adaptation to the present contest of very general facts (see e.g. \cite{JKO,ZP}) and can be easily checked (all proofs have  been provided  in \cite{Fhom}).

\subsection{Square integrable forms}
We define $\nu$ as the Radon measure on $\O \times \bbR^d$ such that 
\begin{equation}\label{labirinto}
 \int d \nu (\o, z) g (\o, z) = \int  d \cP_0 (\o)\int d  \hat \o (z) c_{0,z}(\o) g( \o, z) 
 \end{equation} 
 for any nonnegative Borel function $g(\o,z)$. We point out that $\nu$ has finite total mass since 
$ \nu(\O\times \bbR^d )=\bbE_0[\l_0]<+\infty$.
  Elements of  $L^2 ( \nu)$
 are called \emph{square integrable forms}.

 \smallskip
 
 Given a  function $u:\O_0\to \bbR$,  its gradient  $\nabla u: \O \times \bbR^d \to \bbR$ is defined  as 
 \begin{equation}\label{cantone}
 \nabla u (\o, z):= u (\t_z \o)-u (\o)\,.
 \end{equation}
 If $u$ is defined  $\cP_0$--a.s., then $\nabla u$ is well defined $\nu$--a.s. by Lemma \ref{matteo}.
 If $u $ is bounded and measurable, then $\nabla u \in L^2(\nu)$.
The subspace of \emph{potential forms} $L^2_{\rm pot} (\nu)$ is defined as   the following closure in $L^2(\nu)$:
 \[ L^2_{\rm pot} (\nu) :=\overline{ \{ \nabla u\,:\, u \text{ is  bounded and measurable} \}}\,.
 \]
 The subspace of \emph{solenoidal forms} $L^2_{\rm sol} (\nu)$ is defined as the orthogonal complement of $L^2_{\rm pot} (\nu)$ in $L^2(\nu)$.

\begin{Definition}\label{def_div}
Given a square integrable form $v\in L^2(\nu)$  we define its divergence  ${\rm div} \, v \in L^1(\cP_0)$ as 
\begin{equation}\label{emma}
{\rm div}\, v(\o)= \int d \hat{\o} (z) c_{0,z}(\o) ( v(\o,z)-  v(\t_z \o, -z) )\,.
\end{equation}
\end{Definition}
The r.h.s. of \eqref{emma} is well defined since it corresponds  to an  absolutely convergent series by Lemma \ref{lucertola1}.

 For any  $v \in L^2(\nu)$ and any  bounded and measurable function $u:\O \to \bbR$, it holds (cf. \cite[Lemma 5.4]{Fhom})
\begin{equation}\label{italia}
\int  d \cP_0(\o)   {\rm div} \,v(\o)  u (\o)= - \int d \nu(\o, z) v( \o, z) \nabla u (\o, z) \,.
\end{equation}
As a consequence we have that, given    $v\in L^2(\nu)$, $v\in L^2_{\rm sol}(\nu)$ if and only if ${\rm div}\, v=0$ $\cP_0$--a.s. (cf. \cite[Cor.~5.5]{Fhom}).
We also have (cf. \cite[Lemma 5.8]{Fhom}):
\begin{Lemma}\label{santanna} 
 The  functions $g\in L^2(\cP_0)$ of the form $g ={\rm div}\, v $ with $v\in L^2(\nu)$ are dense in $\{w \in L^2(\cP_0)\,:\, \bbE_0[w]=0\}$.
\end{Lemma}

\subsection{Diffusion matrix} As $\l_2\in L^1(\cP_0)$, 
given $a\in \bbR^d$   the form 
\begin{equation}\label{def_ua}
u_a(\o,z):= a\cdot z
\end{equation}
 is square integrable, i.e. it belongs to $L^2(\nu)$. 
 We note  that the symmetric diffusion matrix $D$  defined in \eqref{def_D} satisfies, for any $a\in \bbR^d$,
 \begin{equation}\label{giallo}
 \begin{split}
q(a):= a \cdot Da  
&= \inf_{ v\in L^2 _{\rm pot}(\nu) } \frac{1}{2} \int d\nu(\o, x) \left(u_a(x)+v(\o,x) \right)^2\\
& =  \inf_{ v\in L^2 _{\rm pot}(\nu) } \frac{1}{2}\| u_a+v \|^2_{L^2(\nu)}=\frac{1}{2} \| u_a+v ^a \|^2_{L^2(\nu)}\,,
\end{split}
 \end{equation}
 where  $v^a=-\Pi u_a$ and 
 $\Pi: L^2(\nu) \to L^2_{\rm pot}(\nu)$ denotes the orthogonal projection of $L^2(\nu)$ on $L^2_{\rm pot}(\nu)$.        It follows 
 easily that  $v^a$ is characterized by the properties
\begin{equation}\label{jung}
v^a \in  L^2_{\rm pot}(\nu)\,, \qquad v^a+u_a\in L^2_{\rm sol}(\nu)\,.
\end{equation}
Moreover it holds (cf. \cite[Section 6]{Fhom}):
 \begin{equation}\label{solare789}
 Da =\frac{1}{2}  \int d\nu (\o, z)  z \bigl(a\cdot z + v^a( \o, z) \bigr)\qquad \forall a \in \bbR^d\,.
\end{equation}

 By \eqref{giallo}
 the   kernel  ${\rm Ker}(q)$ of the quadratic form $q$  is given by 
\begin{equation}\label{rosone}
{\rm Ker}(q):=\{a\in \bbR^d\,:\, q(a)=0\}=\{ a\in \bbR^d\,:\, u_a \in L^2_{\rm pot}(\nu)\}\,.
\end{equation} 
The following result is the analogous of \cite[Lemma 5.1]{ZP}:
\begin{Lemma}\label{rock}\cite[Lemma 6.1]{Fhom}
It holds
\begin{equation}\label{jazz}
{\rm Ker}(q)^\perp= \Big\{  \int  d \nu (\o,z) b(\o, z) z  \,:\, b\in L^2_{\rm sol} (\nu) \Big\}\,. 
\end{equation}
\end{Lemma}
It is simple to check that  Warning \ref{stellina5} and Lemma \ref{rock} imply the following:
 \begin{Corollary}\label{rock_bis} 
 ${\rm Span}\{e_1, e_2, \dots, e_{d_*}\}= \Big\{  \int  d \nu (\o,z) b(\o, z) z  \,:\, b\in L^2_{\rm sol} (\nu) \Big\}$.
  \end{Corollary}

\subsection{The contraction $b(\o,z)\mapsto \hat b(\o) $ and the set  $\cA_1[b]$}\label{sec_cinghia}

\begin{Definition}\label{artic}  Let $b(\o,z): \O _0 \times \bbR^d\to \bbR$ be a Borel function  with $\|b\|_{L^1(\nu)}<+\infty$. We define the Borel function  $c_b:\O_0 \to [0,+\infty]$ as
\begin{equation}\label{magno}
c_b (\o):=  \int d \hat{\o}(z) c_{0,z}(\o)  |b(\o,z)|\,, \end{equation}
 the Borel function $\hat b: \O_0 \to \bbR$ as
\begin{equation}\label{zuppa}
\hat b(\o):= 
\begin{cases}
\int d \hat{\o}(z) c_{0,z}(\o) b(\o,z) & \text{ if } c_b (\o) <+\infty\,,\\
0 & \text{ if } c_b (\o) = +\infty\,,
\end{cases}
\end{equation}
and the  Borel set $\cA_1[b]:= \{ \o\in \O\,:\, c_{b } ( \t _z \o) <+\infty\; \forall z\in \hat \o\}$.
\end{Definition}

We consider  the atomic measures ($\mu^\e_\o$ was introduced in Section \ref{GS})
\begin{equation}\label{melanzane}
  \mu^\e_\o:= \e^d \sum _{x \in \e\hat\o } \d_x\,, \qquad\qquad  \nu ^\e_{\o } : = \sum _{x\in \e\hat \o    }  \sum _{y\in  \e\hat \o   } \e^d  c_{\frac{x}{\e},\frac{y}{\e} }(\o) \d_{( x ,\frac{y-x}{\e})}\,.
   \end{equation}

\begin{Lemma}\label{cavallo}
\cite[Lemma~7.2]{Fhom}  Let $b(\o,z): \O _0 \times \bbR^d\to \bbR$ be a Borel function  with $\|b\|_{L^1(\nu)}<+\infty$. 
 Then 
 \begin{itemize}
 \item[(i)] 
$\| \hat b\|_{L^1(\cP_0)} \leq  \| b\|_{L^1(\nu)}= \| c_b \|_{L^1(\cP_0)} $ and  $\bbE_0[\hat b]= \nu(b)$;
\item[(ii)]
given $\o \in  \cA_1[b]$ and 
   $\varphi \in C_c(\bbR^d)$, it holds\begin{equation}\label{rino}
\int     d\mu_\o^\e (x) \varphi ( x) \hat b (\t_{x/\e} \o) = \int  d \nu_\o^\e (x, z) \varphi(x) b( \t_{x/\e} \o, z) 
\end{equation}
(the  series in the l.h.s. and in the r.h.s.  are absolutely convergent);
\item[(iii)]    $\cP( \cA_1[b])=\cP_0( \cA_1[b])=1$ and 
$ \cA_1[b]$ is translation invariant. 
\end{itemize}
\end{Lemma}

\subsection{The transformation $b(\o,z)\mapsto \tilde b(\o,z) $}\label{hermione}
\begin{Definition}\label{ometto}
Given  a Borel function  $b :\O_0\times \bbR^d\to \bbR$ we set 
\begin{equation}
\tilde b (\o, z):=
\begin{cases}  b (\t_{z} \o, -z)  & \text{ if } z\in \hat \o\,,\\
0 & \text{ otherwise}\,.
\end{cases}
\end{equation}
\end{Definition}
By applying Lemma \ref{matteo} and using Assumption (A3), one gets:
\begin{Lemma}\label{gattonaZZ}  \cite[Lemma~8.2]{Fhom}  Given  a Borel function  $b :\O_0\times \bbR^d\to \bbR$,
it  holds  $\tilde{\tilde b}(\o,z)=b(\o,z)$ if $z\in \hat \o$.  If 
 $b \in L^1(\nu)$, then 
 $ \| b \| _{L^1(\nu)} =  \| \tilde b \| _{L^1(\nu)}$. If 
 $b \in L^2(\nu)$, then 
 $ \| b \| _{L^2(\nu)} =  \| \tilde b \| _{L^2(\nu)}$ and ${\rm div} \,\tilde b= - {\rm div}\, b$. \end{Lemma}

\begin{Definition}\label{naviglio} Let $b: \O_0\times \bbR^d \to \bbR$ be a  Borel function with $\| b\|_{L^1(\nu)}<+\infty$.
 If $\o \in   \cA_1[b]\cap \cA_1[\tilde b]\cap \O_0$,  we set ${\rm div}_* b  (\o) := \hat b (\o) - \hat{\tilde{b}}(\o) \in \bbR$.
\end{Definition}
\begin{Lemma}\label{arranco}\cite[Lemma~8.5]{Fhom}
Let  $b: \O_0\times \bbR^d \to \bbR$ be a Borel function with $\|b\|_{ L^2(\nu)}<+\infty$. Then $\cP_0 ( \cA_1[b]\cap \cA_1[\tilde b]) =1$ and 
${\rm div}_* b  = {\rm div }\,b $ in $L^1(\cP_0)$.
\end{Lemma}
%
%

\begin{Lemma}\label{tav}\cite[Lemma~8.6]{Fhom}
Let  $b: \O_0\times \bbR^d \to \bbR$ be a Borel function with  $\|b\|_{ L^2(\nu)}<+\infty$ and such that  its class of equivalence in $L^2(\nu)$ belongs to $L^2_{\rm sol}(\nu)$. Let 
\be\label{eq_tav}
\cA_d [b]:= \{\o \in \cA_1[b] \cap \cA_1[ \tilde b]\,:\, {\rm div}_* b (\t_z \o) =0 \; \forall z \in \hat \o\}\,.
\en
Then  $\cP( \cA_d [b])=1$ and  $\cA_d [b]$ is translation invariant.
\end{Lemma}

\begin{Lemma}\label{lunetta}\cite[Lemma~8.7]{Fhom}
Suppose that $b: \O_0\times \bbR^d \to \bbR$ is a Borel function with $\|b\|_{ L^2(\nu)}<+\infty$. Take $\o \in   \cA_1[b]\cap \cA_1[\tilde b]$.
 Then for any $\e>0$ and any $u:\bbR^d\to \bbR$ with compact support it holds
\begin{equation}\label{sea}
\int d  \mu^\e_\o (x) u(x) {\rm div}_* b (\t_{x/\e} \o) = - \e \int d  \nu ^\e_\o (x,z) \nabla_\e u(x,z) b ( \t_{x/\e} \o, z) \,.
\end{equation}
\end{Lemma}

 \begin{Lemma}\label{gattonaZ}
 \cite[Lemma~8.3]{Fhom}  $ \,\,$\\
(i) 
 Let $b: \O_0 \times \bbR^d \to [0, +\infty]$ and  $\varphi, \psi:\bbR^d \to [0,+\infty] $  be
  Borel functions. Then, for each $\o \in \O$,  it holds
  \begin{equation}\label{micio2}
 \int d  \nu^\e_\o (x, z) \varphi ( x) \psi (  x+\e z ) b(\t_{x/\e}\o, z)=
  \int d \nu^\e_\o (x, z) \psi (x) \varphi (  x+\e z  ) \tilde{b}(\t_{x/\e}\o, z)\,.
 \end{equation}
 (ii)  Let $b: \O_0 \times \bbR^d\to \bbR$ be a Borel function with $\| b \|_{L^1( \nu )}<+\infty$ and take $ \o \in \cA_1[b]\cap \cA_1[\tilde b]$.  
   Given functions  $\varphi, \psi:\bbR^d \to \bbR $   such that  at least one between $\varphi, \psi $ has  compact support  and the other is bounded,  identity \eqref{micio2} is still valid.
   Given now $\varphi $ with compact support and  $\psi$ bounded, it holds
 \begin{align}
&  \int d  \nu^\e_\o (x, z) \nabla_\e  \varphi   ( x,z) \psi (   x+\e z ) b(\t_{x/\e}\o, z) \nonumber\\
 & \qquad \qquad \qquad  =
  - \int d  \nu^\e_\o (x, z) 
  \nabla _\e \varphi  (x,z) 
  \psi (x) \tilde{b}(\t_{x/\e}\o, z)\,.\label{micio3}
   \end{align}
Moreover,  the above integrals  in  \eqref{micio2}, \eqref{micio3} (under the hypothesis of this  Item (ii)) 
 correspond to absolutely convergent  series and are therefore well defined. 
\end{Lemma}

Recall the set $\cA[g]$ introduced in Prop.~\ref{prop_ergodico} and Definition \ref{sibilla}.
\begin{Lemma}\label{blocco} Suppose that $\o$ belongs to the sets $\cA_1[1]$, $ \cA[\l_0]$, $ \cA_1[ |z|^2  \mathds{1}_{\{ |z| \geq  \ell\}}]$ and $ \cA[  \int d \hat \o (z) c_{0,z}(\o) |z|^2  \mathds{1}_{\{ |z| \geq  \ell\}} ]$ for all $\ell \in \bbN$.
 Then $\forall \varphi \in C_c^2(\bbR^d)$ we have
 \begin{equation}\label{football}
\lim  _{\e \da 0}  \int d \nu^\e_{\o}(x,z)  \bigl[\nabla_\e \varphi (x,z) - \nabla  \varphi (x) \cdot z \bigr]^2   =0\,.
\end{equation}
\end{Lemma}
The above lemma is related to \cite[Lemma 15.2]{Fhom}. We  give the proof,  since we need to isolate the conditions leading to \eqref{football} (which in \cite{Fhom} are assured by the property that $\o$ belongs to the space $\O_{\rm typ}$ in \cite{Fhom}).
\begin{proof} Let $\ell, \phi$ be defined as done before \eqref{paradiso}. 
The  upper bound given by   \eqref{paradiso}  with $\nabla_\e\varphi (x,z)$ replaced by  $\nabla \varphi (x) \cdot z $ is also true. We will apply the above bounds for $|z| \geq  \ell$.  On the other hand, we apply \eqref{mirra} for $|z|<\ell$.
As a result,  we  can bound
\be \label{piano}
 \int d \nu^\e_{ \o}(x,z)  \bigl[\nabla_\e \varphi (x,z) - \nabla  \varphi (x) \cdot z \bigr]^2   \leq 
C(\varphi) [A(\e, \ell )+B(\e, \ell)]\,,
\en
 where (cf. \eqref{micio2}) 
\begin{align*}
 A(\e, \ell ):& =   \int d \nu^\e_{ \o}(x,z) |z|^2  (\phi(x)  + \phi (x+ \e z)  )  \mathds{1}_{\{ |z| \geq  \ell\}} \\
& = 2 \int d \nu^\e_{ \o}(x,z) |z|^2  \phi(x) \mathds{1}_{\{ |z| \geq  \ell\}}  =2 \int d \mu ^\e_{ \o} (x)  \phi(x) h_\ell  (\t_{x/\e} \o) \,,  \\
h_\ell(\o):&= \int d \hat \o (z) c_{0,z}(\o) |z|^2  \mathds{1}_{\{ |z| \geq  \ell\}} \,, \\
 B(\e, \ell):& = \e ^2 \ell^4 \int d \nu^\e_{ \o}(x,z) (\phi(x)  + \phi (x+ \e z)  )  \\
& =2 \e^2  \ell^4 \int d \nu^\e_{ \o}(x,z)  \phi(x)  = 2  \e^2 \ell^4 \int d \mu ^\e_{ \o} (x)  \phi(x) \l_0 (\t_{x/\e} \o)
  \,.
\end{align*}
We now apply Prop. \ref{prop_ergodico}. As $\o \in \cA_1[ |z|^2  \mathds{1}_{\{ |z| \geq  \ell\}}]\cap \cA[h_\ell]$, we conclude that 
    $\lim_{\e \da 0}  \int d \mu ^\e_{ \o} (x)  \phi(x) h_\ell  (\t_{x/\e} \o) = \int dx\,m \phi(x) \bbE_0[h_\ell]$.  Hence  $\lim_{\ell\uparrow  \infty,\e \da 0} A(\e, \ell)=0$ by dominated convergence
   as $\bbE_0[\l_2]<+\infty$.
As $\o \in \cA_1[1]\cap \cA[\l_0]$ the integral $ \int d \mu ^\e_{ \o} (x)  \phi(x) \l_0 (\t_{x/\e} \o)$ converges to $ \int dx\,m \phi(x) \bbE_0[\l_0]$ as $\e \da 0$. As a consequence,  $\lim_{\e \da 0} B(\e, \ell)=0$. Coming back to \eqref{piano} we finally get \eqref{football}.
\end{proof}


\section{The set $\O_{\rm typ}$ of typical environments}\label{sec_tipetto}

Recall  the definitions of the set $\cA[g]$ (cf.~Proposition \ref{prop_ergodico} and Definition \ref{sibilla}) and of the set $\cA_1[g]$ (cf.~ Definition~\ref{artic}).

In the construction of the sets below, we will use the separability of $L^2(\nu)$ and $L^2(\cP_0)$. Since  $(\cN,d)$ is a separable metric space (cf. Section \ref{GS}), the same holds for $(\O, d)$ and $(\O_0, d)$. By \cite[Theorem~4.13]{Br} we then 
 get that the space  $L^p(\cP_0) $ is  separable for $1\leq p <+\infty$.  The separability of $L^2(\nu)$ is proved in \cite[Lemma~9.2]{Fhom}.

\smallskip
\noindent
$\bullet$ {\bf The functional sets $\cG_1,\cH_1$}. 
We fix a countable set $\cH_1$ of Borel functions $b: \O_0\times \bbR^d\to \bbR$ such that 
$\|b \|_{L^2(\nu)}<+\infty$ for any $b \in \cH_1$ and such that   $\{ {\rm div} \,b\,:\, b \in \cH_1\}$ is a dense subset of $\{ w \in L^2(\cP_0)\,:\, \bbE_0[ w]=0\}$ when thought of as set of $L^2$--functions (recall Lemma \ref{santanna}). 
For each $b \in \cH_1$ we define  the Borel function $g_b : \O_0 \to \bbR$ as  (cf. Definition~\ref{naviglio})
\be \label{lupetto}
g_b (\o):= 
\begin{cases}
{\rm div}_* b (\o) & \text{ if } \o \in \cA_1[b]\cap \cA_1[\tilde b]\,,\\
0 & \text{ otherwise}\,.
\end{cases}
\en
Note that by Lemma \ref{arranco} $g_b={\rm div}\, b$, $\cP_0$--a.s.
Finally we set $\cG_1:=\{ g_b \,:\, b \in \cH_1\}$. 

\smallskip

\noindent
$\bullet$ {\bf The functional sets $\cG_2,\cH_2 $}.  We fix a countable set $\cG_2$  of  bounded Borel functions 
$g: \O_0\to \bbR$ such that  the set $\{ \nabla g\,:\, g \in \cG_2\}$, thought in $L^2(\nu)$,  is  dense in $L^2_{\rm pot}(\nu)$
 (this is possible by the definition of $L^2_{\rm pot}(\nu)$).  We  define $\cH_2$ as the set of Borel functions $h:\O_0 \times \bbR^d\to \bbR$ such that $h=\nabla g$ for some $g\in \cG_2$.

\smallskip

\noindent
$\bullet$ {\bf The functional set $\cW$}. We fix a  countable set $\cW$ of Borel functions $b:\O_0\times \bbR^d\to \bbR$ such that, thought of as subset of $L^2(\nu)$, $\cW$  is dense in $L^2_{\rm sol} (\nu)$.  By   Lemma~\ref{gattonaZZ},  $\tilde b \in L^2_{\rm sol}(\nu)$ for any $b \in L^2_{\rm sol}(\nu)$. Hence, at cost to enlarge $\cW$, we assume that $\tilde b \in \cW$ for any $b \in \cW$ (recall Definition \ref{ometto}).

\begin{Definition}[Definition of the functional  set $\cG$] 
We define $\cG$ as the union of the following countable sets of Borel functions on $\O_0$, which are $\cP_0$--square integrable:
$\{1\}$,  $\cG_1$, $\cG_2$ and $\{ u_{b,i} \mathds{1}( |u_{b,i}|\leq M)\}$ with  $ b \in \cW$, $i \in \{1,\dots, d\}$, $M\in \bbN $ and $u_{b,i}(\o):=
\int d \hat \o (z)  c_{0,z}(\o) z_i b(\o,z)$.
\end{Definition}
\begin{Definition}[Definition of the functional  set $\cH$] We define $\cH$  as the union of the following countable sets  of Borel functions on $\O_0\times \bbR^d$, which are $\nu$--square integrable: $\cH_1$, $\cH_2$, $\cW$, $\{ (\o,z) \mapsto z_i \,:\, 1\leq i \leq d\}$.
\end{Definition}

Recall the transformation $b \mapsto \hat{b}$ given in Definition \ref{artic} and the parameter $\a\in (0,1)$ appearing in Assumption (A6).
\begin{Definition}\label{amen}
The set $\O_{\rm typ}\subset \O$  of typical environments is  the intersection of the following sets:
\begin{itemize}
\item  $\cA[g g'] $ for all $g,g'\in \cG$ (recall that $1\in \cG$); 
\item $\cA_1[b b'] \cap \cA[ \widehat{ b b'}]$ as $b,b'\in \cH$;
\item    $\O_2$ (cf.~ Lemma \ref{paletta});
\item  $\cA_1[ |z|^k]\cap \cA[\l_k]  $ for $k=0,2$;
\item $\cA[   \int d \hat \o (z) c_{0,z} (\o) |z|^2 \mathds{1}_{|z|\geq n}]$ for all $n \in \bbN$;
\item  $\cA_1[ c_{0,z}(\o)^\a] \cap  \cA\left[ \int d \hat \o (z) c_{0,z}(\o) ^\a  \mathds{1}_{\{ |z| \geq  n\}} \right]$ for all $n\in \bbN$;
\item $\cA_1[b]\cap \cA_1[\tilde b]\cap \cA_1[ b^2] \cap \cA_1[ \tilde{b}^2] $ for all $b\in \cH$;
\item $
  \cA[ \widehat{\;b^2\;}]\cap  \cA[ \widehat{\;\tilde{b}^2\;}] \cap  \cA  [\widehat{\,| b|\,} ]\cap  \cA [\widehat{ \, |\tilde b|\,}]$ for all $b\in \cH$;
\item $\cA_1[ \tilde b(\o,z) z_i]$ for $1\leq i \leq d$ for all $b\in \cW$;
\item $\cA[ u _{ b, i, M}]$ for all  $b \in \cW$, $1\leq i \leq d$ and $M\in \bbN$, where 
$u _{ b,i,M}:= | u_{ b ,i}| \mathds{1}\bigl(| u_{ b ,i}|\geq M\bigr)  $ and $u _{  b,i }(\o):= \int d\hat \o (z)
c_{0,z}(\o) z_i  {b} (\o, z)$ (see definition of $\cG$);
\item $ \cA_1[c_{0,z}(\o)^\a z_1^2]\cap \cA[ \int d\hat \o(z) c_{0,z}(\o)^\a z_1^2 ]$;
\item $\cA_d [b]$ for all $b\in \cW$ (recall \eqref{eq_tav}).
\end{itemize}
%
\end{Definition}
As $\l_0, \l_1\in L^1(\cP_0)$,  due to \eqref{sorpresina}, \eqref{zarina} and our definition of $\cG$, $\cH$, $\cW$, the  
sets listed in Definition \ref{amen} are well defined (recall in particular Lemmata \ref{cavallo}, \ref{gattonaZZ}, \ref{arranco}).  As these sets are translation invariant with full $\cP$-measure (see   Proposition \ref{prop_ergodico}, Lemma \ref{cavallo} and Lemma \ref{tav}), the same holds for $\O_{\rm typ}$. 

%


\section{Weak/strong convergence and  2-scale convergence}\label{anatre12}
Recall   $\mu^\e_{\o,\L} $ and $\nu^\e_{\o,\L}$  given in \eqref{atomiche}.
Recall $\mu^\e_{\o} $ and $\nu^\e_{\o}$  given in \eqref{melanzane}.
 We also define
\be\label{acqua}
 \mu^\e_{\o,S}:= \e^d \sum _{x \in \e\hat \o  \cap S} \d_x\,,\qquad  \nu ^\e_{\o,S} : = \sum _{x\in \e\hat \o  \cap S }  \sum _{y\in  \e\hat \o  \cap S } \e^d  
 c_{\frac{x}{\e},\frac{y}{\e} }(\o) \d_{( x ,\frac{y-x}{\e})}\,.
  \end{equation}
In what follows,  $\D$ equals $S$ or $\L$.
\subsection{Weak/strong convergence}\label{sec_weak_strong}

\begin{Definition}\label{debole_forte} Fix $\o\in \O$ and a 
 family of $\e$--parametrized  functions $v_\e \in L^2( \mu^\e_{\o,\D})$.

 $\bullet$  We say  that the family $\{v_\e\}$  \emph{converges weakly} to the function  $v\in L^2( \D, m dx)$, and write  $v_\e \rightharpoonup v$, if the family $\{v_\e\}$ is bounded  (i.e. 
$\limsup  _{\e \da 0} \| v_\e\| _{L^2(\mu^\e_{\o,\D}) } <+\infty)$
and 
\begin{equation}\label{deboluccio}
\lim _{\e \da 0} \int d  \mu^\e _{\o,\D}  (x)   v_\e (x) \varphi (x)= 
\int_\D dx\, m   v(x) \varphi(x) \end{equation}
for all  $\varphi \in C_c(\D)$.

$\bullet$  We  say that  the  family $\{v_\e\}$  \emph{converges strongly} to  $v\in L^2(\D, m dx)$, and write $v_\e\to v$,  if  $\{v_\e\}$ is bounded and   it holds
\begin{equation}\label{fortezza}
\lim _{\e \da 0} \int  d  \mu^\e _{\o ,\D} (x)   v_\e (x) g_\e (x)= 
\int_\D dx\,  m v(x) g(x) \,,\end{equation}
for any family of functions $g_\e \in  L^2(\mu^\e_{\o,\D} )$ weakly converging to $g\in L^2( \D, m dx) $.
\end{Definition}
Trivially,  strong convergence implies weak convergence. 

\begin{Remark}\label{utilino} Given $v_\e$ and $v$ as in Definition \ref{debole_forte}, we have 
that $v_\e \to v $ if $v_\e \rightharpoonup v$ and $\lim  _{\e \da 0} \| v_\e\| _{L^2(\mu^\e_{\o,\D}) } =
\|v\|_{L^2(\D, m dx)}$ (cf. the proof of \cite[Prop.~1.1]{Z}). \end{Remark}

%
%

\subsection{Weak 2-scale convergence}\label{sec_222}

  
\begin{Definition}\label{priscilla}
Fix   $\tilde \o\in \O_{\rm typ}$, an $\e$--parametrized  family of functions  $v_\e \in L^2( \mu^\e_{\tilde \o,\D })$ and  a function $v \in L^2 \bigl(\D \times \O, m dx \times \cP_0\bigr)$.
 We say that \emph{$v_\e$ is weakly 2-scale convergent to $v$}, and write 
$v_\e \stackrel{2}{\rightharpoonup} v$, 
if the family $\{v_\e\}$ is bounded, i.e.
$
 \limsup_{\e\downarrow 0}  \|v_\e\|_{L^2(\mu^\e_{\tilde \o,\D})}<+\infty$, 
 and 
\begin{equation}\label{rabarbaro}
\lim _{\e\downarrow 0} \int d \mu _{\tilde \o, \D }^\e (x)  v_\e (x) \varphi (x) g ( \t _{x/\e} \tilde{\o} ) =\int d\cP_0(\o)\int _\D dx\,  m v(x, \o) \varphi (x) g (\o)  \,,
\end{equation}
for any $\varphi \in  C_c (\D)$ and any $g \in\cG$.
\end{Definition}
One can define also the strong 2-scale convergence, but we will not need it in what follows.
As  $\tilde \o \in \O_{\rm typ}\subset  \cA[g]$ for all $g\in \cG$, by Proposition \ref{prop_ergodico} one  gets  that $v_\e \stackrel{2}{\rightharpoonup} v$ where $v_\e:= \varphi \in L^2(\mu^\e_{\tilde \o,\D})$ and $v:=\varphi \in L^2(\D, m dx)$ for any $\varphi \in C_c(\D)$.
%
%
 \smallskip
 
It is standard to prove the following fact  by using  the first item  in Definition~\ref{amen} (cf.~\cite[Prop.~2.2]{Z}, \cite[Lemma~5.1]{ZP} and in particular \cite[Lemma~10.5]{Fhom}):

\begin{Lemma}\label{compatto1}
 Let $\tilde \o\in \O_{\rm typ}$.    Then, given a bounded family of functions $v_\e\in L^2 ( \mu^\e_{\tilde{\o},\D})$,  there exists a subsequence $\{v_{\e_k}\}$ such that
 $ v_{\e_k}
  \stackrel{2}{\rightharpoonup}   v $ for some 
 $  v \in L^2(\D \times \O, m dx \times \cP_0 )$ with $\|  v\|_{   L^2(\D\times \O, m dx \times \cP_0 )}\leq \limsup_{\e\da 0} \|v_\e\|_{L^2(\mu^\e_{\tilde{\o},\D})}$.
 \end{Lemma}

Recall the definition of the measure $\nu$ given in \eqref{labirinto}.
 
 \begin{Definition}
Given  $\tilde \o\in \O_{\rm typ}$, an $\e$--parametrized  family of functions $w_\e \in L^2( \nu ^\e_{\tilde \o, \D })$ and  a function  $w \in L^2 \bigl( \D\times \O\times \bbR^d\,,m  dx \times d\nu\bigr)$, we say that \emph{$w_\e$ is weakly 2-scale convergent to $w$}, and write $  w_\e \stackrel{2}{\rightharpoonup}w   $,  if   $\{w_\e\}$ is bounded
 in $L^2( \nu ^\e_{\tilde \o,\D})$, i.e. $
 \limsup_{\e \da 0} \|  w_\e\|_{L^2( \nu ^\e_{\tilde \o,\D})}<+\infty$,
   and
\begin{multline}\label{yelena}
 \lim _{\e\downarrow 0} \int   d \nu_{\tilde \o,\D}^\e (x,z) w_\e (x,z ) \varphi (x) b ( \t _{x/\e} \tilde{\o},z )\\
=\int_\D  dx \,m\int d \nu (\o, z) w(x, \o,z ) \varphi (x) b (\o,z )  \,,
\end{multline}
for any    $\varphi \in C_c(\D)$  and  any  $b \in \cH $. 
\end{Definition}

 \smallskip
 
It is standard to prove the following fact  by using the second item  in Definition~\ref{amen} (cf.~\cite[Lemma~10.7]{Fhom}):

\begin{Lemma}\label{compatto2} Let $\tilde \o\in \O_{\rm typ}$.  Then, given a bounded family of functions $w_\e\in L^2 ( \nu^\e_{\tilde{\o},\D })$,  there exists a subsequence $\{w_{\e_k}\}$   such that  $w_{\e_k} \stackrel{2}{\rightharpoonup} w$ for some 
 $ w \in L^2(  \D\times \O\times \bbR^d\,,\, m dx \times \nu )$ with  $\|  w\|_{   L^2(\D\times \O\times \bbR ^d\,,\, m dx \times \nu )}\leq \limsup_{\e\da 0} \|w_\e\|_{L^2( \nu^\e_{\tilde{\o},\D})}$.
\end{Lemma}

\section{$2$-scale limits of uniformly bounded functions} 
\label{sec_bike}


We fix $\tilde \o \in \O_{\rm typ}$. The domain $\D$ below  can be $\L,S$. We  consider a family of functions $\{f_\e\}$  with $f_\e: \e \widehat{\tilde \o} \cap S  \to \bbR $
such that 
\begin{align}
& \limsup_{\e\da 0} \|  f_\e\|_\infty <+\infty\,,\label{re1}\\
&\limsup_{\e \da 0}   \| f_\e \|_{L^2(  \mu^\e_{\tilde \o, \D})} <+\infty \,,\label{re2} \\
& \limsup_{\e \da 0}   \|\nabla_\e  f_\e \|_{L^2(\nu^\e_{\tilde \o,\D})} < +\infty\,.\label{re3}
\end{align}

Due to Lemmata \ref{compatto1} and \ref{compatto2}, along a subsequence $\{\e_k\}$ we have 
\begin{align}
&L^2(  \mu^\e_{\tilde \o, \D})\ni  f_\e \stackrel{2}{\toup} v  \in L^2(\D \times \O,  m dx \times \cP_0 ) \,,\label{totani1}\\
& L^2(\nu^\e_{\tilde \o,\D})\ni \nabla_\e{f}_\e \stackrel{2}{\toup} w\in   L^2(\D\times \O\times \bbR^d  \,,\, m dx \times \nu ) \,, \label{totani2}
\end{align}
for suitable functions $v,w$. 
\begin{Warning}
 In this section (with exception of Lemma \ref{lunghissimo} and Claim \ref{mengoni}), when taking the limit $\e\da 0$, we understood that $\e$ varies along the subsequence  $\{\e_k\}$ satisfying \eqref{totani1} and \eqref{totani2}.  We set 
 $  \bar f_\e(x):=0 $ for $ x\in \e \widehat{\tilde  \o} \setminus S $.
\end{Warning}

The structural results presented below (cf. Propositions \ref{prova2} and \ref{oro})  correspond to  a general strategy in homogenization by 2-scale convergence (see Propositions 12.1 and 14.1 in \cite{Fhom}, Lemmata 5.3 and  5.4 in \cite{ZP}, Theorems 4.1 and 4.2 in \cite{Z}).
 Condition \eqref{re1} would not be strictly necessary, but it allows important technical simplifications, and in particular it allows to avoid the cut-off procedures developed in  \cite[Sections 11,13]{Fhom} in order to deal with the long jumps in the Markov generator  \eqref{degregori}. We will apply Propositions \ref{prova2} and \ref{oro} only to the following  cases: $\D=\L$ and $f_\e={ V_\e}$; $\D=S$ and $f_\e=V_\e-\psi$. In both cases  \eqref{re1}, \eqref{re2} and \eqref{re3}  are satisfied by   Remark \ref{fiorellino3} and Lemma \ref{paletta}. 

\smallskip
 In what follows we will use the following control on long filaments (recall \eqref{melanzane}):
 
 \begin{Lemma}\label{lunghissimo}  Given $\tilde \o \in \O_{\rm typ}$,  $\ell>0$ and $\varphi \in C_c(\bbR^d)$,  it holds
 \be\label{lionello}
\lim_{\e\da 0}  \e^{-2} \int d\nu ^\e_{\tilde \o} (x,z) \,|\varphi(x)| \, \mathds{1}(|z|\geq \ell/\e)=0\,.
\en
\end{Lemma}
\begin{proof} Let $\a\in (0,1)$ be as in  (A6).
We set  
$\k(t):=\sup_{\o \in \O_0, |z| \geq t} c_{0,z}(\o) ^{1-\a} $ and
$
h_{\a,n}(\o):= \int d \hat \o  (z)  c_{0,z}(\o)^\a \mathds{1}(|z|\geq n)$ 
for $n \in \bbN$.
For $\ell /\e\geq n$, we can bound the l.h.s. of \eqref{lionello} by 
\be\label{yom}
\begin{split}
& \e^{-2}  \int d \mu^\e_{\tilde \o} (x) |\varphi(x)|\int d \widehat{\t_{x/\e} \tilde \o}(z)  c_{0,z}(\t_{x/\e}\tilde \o)\mathds{1}(|z|\geq \ell/\e)\\
&  \leq\e^{-2} \k(\ell/\e)  \int d \mu^\e_{\tilde \o} (x)| \varphi(x) |  \int d \widehat{\t_{x/\e} \tilde \o}(z)  c_{0,z}(\t_{x/\e}\tilde \o)^\a \mathds{1}(|z|\geq n)\\
& =\e^{-2} \k(\ell/\e)   \int d \mu^\e_{\tilde \o} (x) |\varphi(x)|  h_{\a,n}(\t_{x/\e} \tilde \o)\,.
\end{split}
\en
By \eqref{downtown} we have $\limsup_{\e \da 0}\e^{-2}\k(\ell/\e) <+\infty$. Since $\tilde \o \in \O_{\rm typ}\subset \cA_1[ c_{0,z}(\o)^\a] \cap  \cA[h_{\a,n}]$,   we have $\int d \mu^\e_{\tilde \o} (x) |\varphi(x)|\ h_\a(\t_{x/\e} \tilde \o) \to \int dx \,m |\varphi(x)|  \bbE_0[h_{\a,n}] $ as $\e\da 0$. By taking  the limit $n\to \infty$ we get \eqref{lionello} due to \eqref{sorpresina}.
\end{proof}

\begin{Proposition}\label{prova2}
 For $dx$--a.e. $x\in \D$, the map  $ v(x,\o)$ given in \eqref{totani1}  does not depend on $\o$.
\end{Proposition}
\begin{proof} 
Recall the definition of the functional sets $\cG_1, \cH_1$ given in Section \ref{sec_tipetto}.
We  claim that $\forall \varphi \in C^1_c(\D)$ and $\forall \psi \in \cG_1$ it holds 
\begin{equation}\label{chiavetta}
\int _\D dx \,m \int d\cP_0 (\o) v (x,\o) \varphi (x) \psi(\o)=0\,.
\end{equation}
Before proving our claim, let us explain how it leads to the thesis. Since $\varphi $ varies among $C^1_c(\D)$ while $\psi$ varies in a countable set,  \eqref{chiavetta} implies that,  $dx$--a.e. on $\D$,  $ \int \cP_0 (\o) v (x,\o)\psi(\o)
=0$ for any $\psi \in \cG_1$. We conclude that,  $dx$--a.e. on $\D$, $v(x,\cdot)$ is orthogonal in $L^2(\cP_0 )$ to $\{ w \in L^2(\cP_0 )\,:\, \bbE_0[ w]=0\}$ (due to the density of $\cG_1$), which is equivalent to the fact that $v(x,\o)= \bbE_0[ v(x, \cdot)]$ for $\cP_0$--a.a. $\o$.

It now remains to prove \eqref{chiavetta}. 
 We first note  that, by   \eqref{rabarbaro},  \eqref{totani1} and since  $\tilde \o \in \O_{\rm typ}$ and $\psi \in \cG_1\subset \cG$,
\begin{equation}\label{nord1}
\text{l.h.s. of }\eqref{chiavetta}= \lim_{\e\da 0} \int d \mu^\e_{\tilde \o,\D} (x) f_\e(x) \varphi (x) \psi( \t_{x/\e}\tilde \o) \,.
\end{equation}
Let us take $\psi= g_b$  with $b \in \cH_1$ as in \eqref{lupetto}. 
 By Lemma \ref{lunetta}  and since $\tilde \o \in \O_{\rm typ}\subset \cA_1[b]\cap \cA_1[\tilde b]$, 
we have
\be\label{incubi}
\begin{split}
\int d \mu^\e_{\tilde \o,\D} (x) f_\e(x) \varphi (x) \psi( \t_{x/\e}\tilde \o) & =\int d\mu^\e_{\tilde \o} (x) \bar f_\e(x) \varphi (x) \psi( \t_{x/\e}\tilde \o) \\
&=
- \e \int d  \nu ^\e_{\tilde \o} (x,z) \nabla_\e( \bar f_\e  \varphi)(x,z) b ( \t_{x/\e} \tilde \o, z)\,.
\end{split}
\en
As usual,   we think $C_c(\D)\subset  C_c(\bbR^d)$ and we keep the same notation for  $\varphi$ thought in $C_c(\bbR^d)$. By \eqref{leibniz} we have
\be \label{nord2}
 -\e \int d\nu ^\e_{\tilde \o} (x,z) \nabla_\e( \bar f_\e \varphi ) (x,z) b (\t_{x/\e} \tilde{\o}, z)= -\e C_1(\e)+ \e C_2(\e)\,,
\en
where
 \begin{align*}
& C_1(\e):= \int d\nu ^\e_{\tilde \o} (x,z) \nabla_\e \bar f_\e (x,z) \varphi(  x) b (\t_{x/\e} \tilde{\o}, z)\,,\\
&  C_2(\e):=  \int d \nu ^\e_{\tilde \o} (x,z) \bar f_\e ( x+\e z  ) \nabla_\e \varphi( x,z) b (\t_{x/\e} \tilde{\o}, z)   \,.
\end{align*}

Due to \eqref{nord1}, \eqref{incubi} and \eqref{nord2}, to get \eqref{chiavetta} we only need to show that $\lim_{\e \da 0} \e C_1(\e)=0$ and $\lim_{\e \da 0} \e C_2(\e)=0$.

We start with $C_1(\e)$. By Schwarz inequality  and since $\tilde \o \in \O_{\rm typ}\subset   \cA _1 [ b^2  ] $
\bes
|C_1(\e)|\leq  \Big[ \int d\nu ^\e_{\tilde \o} (x,z) |\varphi(x)| \nabla _\e \bar f _\e (x,z) ^2\Big]^{1/2} 
\Big[ \int \mu^\e_{\tilde \o} (x) | \varphi(x) |  \widehat{\, b^2\, }(\t_{x/\e} \tilde\o)
\Big]^{1/2}\,.
\ens
Since $\tilde \o \in \O_{\rm typ}\subset   \cA _1 [ b^2  ] \cap \cA [  \widehat{\,b^2\,}  ]$, the last integral in the  r.h.s. converges to a finite constant as $\e \da 0$. 
It remains to prove that $ \int d\nu ^\e_{\tilde \o} (x,z) |\varphi(x)| \nabla _\e \bar f _\e (x,z) ^2 $
remains bounded from above as $\e \da 0$.  
We call $\ell $ the distance between  the support  of $\varphi$ (which is contained in $\D$ as $\varphi \in C^1_c(\D)$) and $\partial \D$.  Then, between the pairs $(x,z) $ with $x+\e z \not \in S$  contributing to the above integral, only the pairs 
$(x,z)$ such that $x\in \D$ and $|z| \geq \ell/\e$ can give a nonzero contribution.
In both cases $\D=\L$ and $\D=S$ we can estimate 
\be\label{olivia}
\begin{split}
\int d\nu ^\e_{\tilde \o} (x,z) |\varphi(x)| & \nabla _\e \bar f _\e (x,z) ^2  
 \leq \int d\nu ^\e_{\tilde \o,\D} (x,z) |\varphi(x)| \nabla _\e \bar f _\e (x,z) ^2  \\
&+
\int d\nu ^\e_{\tilde \o} (x,z)| \varphi(x) |\nabla _\e \bar f _\e (x,z) ^2\mathds{1}(|z| \geq \ell /\e)\,.
\end{split}
\en
The first addendum in the r.h.s. of \eqref{olivia} is bounded due to \eqref{re3}. The second addendum goes to zero due to  \eqref{re1} (implying that $|\nabla_\e \bar f|\leq C/\e$ for small $\e$) and  Lemma \ref{lunghissimo}. Hence the l.h.s. of \eqref{olivia} remains bounded as $\e\da 0$.
This completes the proof that $\lim_{\e \da 0} \e C_1(\e)=0$.

 We move to $C_2(\e)$. Let $\phi$ be  as in \eqref{paradiso}. 
Using  \eqref{paradiso} and \eqref{re1}, and afterwards Lemma \ref{gattonaZ}--(i), for some $\e$--independent constants $C$'s  (which can change from line to line), for $\e$ small   we can bound 
\be \label{crepe}
\begin{split}
& |C_2(\e)|  \leq C
 \int d \nu ^\e_{\tilde \o} (x,z)\bigl | \nabla_\e \varphi( x,z) b (\t_{x/\e} \tilde{\o}, z)\bigr| \\
& \leq C \int d   \nu ^\e _{\tilde \o} (x,z) |z|\,|  b (\t_{x/\e}\tilde \o, z) | \bigl( \phi(x)+ \phi(x+\e z) \bigr) \\
 &\leq  C \int d  \nu ^\e_{\tilde \o} (x,z) 
 \phi(x) |z| (| b|+|\tilde b|)  (\t_{x/\e} \tilde{\o}, z) \\
 & \leq C \Big[ \int d  \nu ^\e_{\tilde \o} (x,z) 
 \phi(x)  |z| ^2 \Big]^{1/2} \Big[ 2 \int d  \nu ^\e_{\tilde \o} (x,z) 
 \phi(x)  (b^2 +\tilde b^2 )  (\t_{x/\e} \tilde{\o}, z)  \Big]^{1/2}\,.
 \end{split}
 \en
 The first integral in the last line of \eqref{crepe} equals $\int d \mu ^\e _{\tilde \o} (x)\phi(x)\l_2 (\t_{x/\e}\tilde \o)$. Since $\tilde \o \in \O_{\rm typ}\subset \cA_1[|z|^2] \cap \cA[\l_2]$, this integral converges to  a finite constant  as $\e\da0$.  The second  integral in the last line of  \eqref{crepe} equals 
 \begin{equation}\label{nuoto}
 \int d  \mu ^\e_{\tilde \o} (x) 
 \phi(x)  (\widehat{\, b^2\,} +\widehat{\,\tilde b^2\,} )  (\t_{x/\e} \tilde{\o})
 \end{equation}
 as $\tilde \o \in \O_{\rm typ}\subset \cA_1[b^2 ]\cap  \cA_1[ \tilde{b}^2 ]$. 
 Since    $\tilde \o \in \O_{\rm typ}\subset \cA[  \widehat{\, b^2\,}   ]\cap \cA[\widehat{\,\tilde b^2\,} ] $, the integral \eqref{nuoto} converges
 to a finite constant. This implies  that $\lim_{\e \da 0} \e C_2(\e)=0$.
 \end{proof}

Due to  Proposition \ref{prova2} we can  write $v(x)$ instead of $v(x,\o)$, where $v$ is given by \eqref{totani1}. Recall the index $d_*$ introduced in Warning \ref{stellina5} and recall \eqref{aria7}.

\begin{Proposition}\label{oro} Let $v$ and $w$ be as in \eqref{totani1} and \eqref{totani2}. Then it holds:
\begin{itemize}
\item[(i)]
 $v $ has weak derivatives $\partial_j v \in L^2(\D,dx)$ for  $j: 1\leq j \leq d_*$;  
\item[(ii)] $w(x,\o,z)= \nabla_* v (x) \cdot z   + v_1 (x,\o,z)$,
where $v_1\in L^2\bigl( \D, dx;L^2_{\rm pot} (\nu)\bigr )$.
\end{itemize}
\end{Proposition}

We stress that $L^2\bigl( \D, dx;L^2_{\rm pot} (\nu)\bigr )$ denotes the space of square integrable maps $f:\D \to L^2_{\rm pot} (\nu)$,
where $\D$ is endowed with the Lebesgue measure.

%
%
%

\begin{proof}
Given a square integrable form $b$, we define $\eta_b := \int d\nu (\o,z) z b(\o,z)$.
Note that $\eta_b$ is well defined since both $b$ and the map $(\o,z) \mapsto z$ are in $ L^2(\nu)$ (for the latter use that $\bbE_0[\l_2]<+\infty$).
We observe that $\eta_b=- \eta_{\,\tilde b}$  by Lemma \ref{lucertola1}  with $k(\o,\o') := z c_{0,z} (\o) b (\o,z)$ if $\o'$ can be written as $\t_z \o$ with $z\in \hat \o$ and $k(\o,\o'):=0$ otherwise (the function $k$ is well defined $\cP_0$--a.s. due to Assumption (A3)).
We claim that 
 for   each  solenoidal form $b  \in L^2_{\rm sol}(\nu) $ and each  function $\varphi \in C^2_c(\D)$, it holds 
 \begin{equation}\label{kokeshi}
 \int _\D dx \,  m \varphi(x) \int  d \nu (\o ,z) w(x,\o,z) b(\o, z) =  
  -\int  _\D dx\,m   v(x) \nabla \varphi(x) \cdot \eta_b\,.
\end{equation}

Before proving   \eqref{kokeshi} we show how to conclude the proof of Proposition~\ref{oro}. We  start with Item (i). Due to Corollary \ref{rock_bis}
there are solenoidal forms $b_1,b_2,\dots, b_{d_*}$ such that $\eta_{b_1}, \eta_{b_2}, \dots, \eta_{b_{d_*}}$ equals $e_1, e_2, \dots, e_{d_*}$. Given $1\leq i \leq d_*$ consider the measurable function 
  \begin{equation}
  g_i (x):=  \int   d \nu (\o ,z) w(x,\o,z) b_i(\o, z)\,, \qquad x\in \D\,.
  \end{equation}
We have that $g_i \in L^2(\D,dx) $. Indeed, by Schwarz inequality and  since  $w\in   L^2(\D\times \O\times \bbR^d  \,,\, dx \times \nu )$, we can bound
\begin{multline}\label{greta}
\int _\D g_i(x) ^2 dx =
   \int _\D dx \left[ \int  d \nu (\o ,z) w(x,\o,z) b_i(\o, z)\right]^2\\
\leq  \|b _i \|^2_{L^2(\nu)}\int_\D dx   \ \int  d \nu (\o ,z) w(x,\o,z) ^2 <\infty\,.
\end{multline}
   Moreover, we have that $\int_\D dx \,\varphi(x) g_i(x) =  -\int _\D  dx\, v(x) \partial_i \varphi(x) $  by \eqref{kokeshi} and since $\eta_{b_i}=e_i$. This proves that    $\partial_i v(x)= 
- g_i(x) \in L^2(\D,dx) $, $\partial_i v$ being the weak derivative of $v$ w.r.t. the $i$--th coordinate.
  This concludes the proof of Item (i).
  
  We move to Item (ii) (always assuming 
  \eqref{kokeshi}). By Item (i) and Corollary \ref{rock_bis}  we can  replace  the r.h.s. of \eqref{kokeshi}  by $ \int _\D dx\, m \, (\nabla_* v(x) \cdot \eta_b ) \varphi(x)$. Hence   \eqref{kokeshi} can be rewritten as 
 \begin{equation}\label{cocco}
 \int _\D dx  \varphi(x) \int   d \nu (\o ,z)
 \left[ w(x,\o,z) -\nabla_* v(x)  \cdot z
 \right] b(\o, z) =  
  0\,.
  \end{equation}
By the arbitrariness of $\varphi$ we conclude that $dx$--a.s.  on $\D$
\begin{equation}\label{criceto13}
\int  d \nu (\o ,z)
 \left[ w(x,\o,z) -\nabla_* v(x)  \cdot z
 \right] b(\o, z) =  
  0\,, \qquad \forall b \in L^2 _{\rm sol}(\nu)\,.
\end{equation}
Let us now show that the map 
$w(x,\o,z) -\nabla_* v(x) \cdot z $ belongs to $L^2(\D,dx; L^2(\nu) )$. Indeed,  we have 
$
\int _\D dx \|w(x, \cdot, \cdot)\|_{L^2(\nu) }^2 =
 \| w\|^2_{ L^2(\D\times \O, dx \times \nu)}<+\infty$ and also
 \begin{equation}
\int _\D dx \|\nabla_* v(x) \cdot z \|_{L^2(\nu) }^2 \leq \int _\D dx | \nabla_* v(x) |^2  \int d\nu (\o, z) |z|^2 
<\infty\,,
\end{equation}
 where the last bound follows from the fact that  $\nabla_* v\in L^2(\D,dx)$ (see Item (i)) and  that $\bbE_0[\l_2]<+\infty $.

As the map 
$w(x,\o,z) -\nabla_* v(x) \cdot z $ belongs to $L^2(\D,dx; L^2(\nu) )$, for $dx$--a.e. $x$ in $\D$ we have that the map 
  $(\o,z) \mapsto w(x,\o,z) -\nabla_* v(x) \cdot z$ belongs to $ L^2(\nu)$ and therefore, by \eqref{criceto13}, to   $ L^2_{\rm pot} (\nu)$.
This concludes the proof of Item (ii).

\smallskip
It remains to prove \eqref{kokeshi}. Since both sides of \eqref{kokeshi} are continuous as functions of $b \in L^2_{\rm sol}(\nu)$, it is enough to prove it for $b\in \cW$. 
  Since $\tilde \o \in \O_{\rm typ}$,  along $\{\e_k\}$ it holds $\nabla_\e  f_\e \stackrel{2}{\toup} w$ as in  \eqref{totani2} and since  $b \in \cW \subset \cH$ (cf.~\eqref{yelena}) we can write 
  \be\label{bruna1}
  \begin{split}
  \text{l.h.s. of }\eqref{kokeshi}&= \lim _{\e \da 0} \int d \nu ^\e _{\tilde \o, \D}(x,z) \nabla_\e f_\e(x,z) \varphi (x) b ( \t_{x/\e} \tilde \o, z)\\
  &= \lim _{\e \da 0} \int d \nu ^\e _{\tilde \o, \D}(x,z) \nabla_\e \bar f_\e(x,z) \varphi (x) b ( \t_{x/\e} \tilde \o, z)\,.
  \end{split}
  \en
  Since $b \in \cW\subset L^2_{\rm sol}(\nu)$ and   $\tilde \o \in \O_{\rm typ}\subset \cA_d[b] $, from     Lemma  \ref{lunetta}  we get 
    \[
  \int d  \nu ^\e _{\tilde \o}(x,z) \nabla_\e ( \bar f_\e \varphi ) (x,z)  b ( \t_{x/\e} \tilde \o, z)=0\,.
  \]
  Above we used the natural inclusion $C_c(\D) \subset C_c(\bbR^d)$.
Using the above identity and  \eqref{leibniz}, we get 
\be\label{bruna2}
\begin{split}
\int d  \nu ^\e _{\tilde \o}(x,z) &\nabla_\e \bar f_\e (x,z) \varphi (x) b ( \t_{x/\e} \tilde \o, z)\\
&=- \int d  \nu ^\e _{\tilde \o}(x,z)  \bar f_\e (x+ \e z) \nabla_\e \varphi (x,z) b ( \t_{x/\e} \tilde \o, z)\,.
\end{split}
\en
As a byproduct of \eqref{bruna2} and 
  \eqref{micio3} in Lemma 
\ref{gattonaZ}--(ii), we get
\be\label{mammina76}
\int d  \nu ^\e _{\tilde \o}(x,z) \nabla_\e \bar f_\e (x,z) \varphi (x) b ( \t_{x/\e} \tilde \o, z)=
 \int d  \nu ^\e _{\tilde \o}(x,z)\bar   f_\e (x) \nabla_\e \varphi (x,z) \tilde b ( \t_{x/\e} \tilde \o, z)
 \,.
 \en
 By combining \eqref{bruna1} and \eqref{mammina76} we therefore have that 
\be\label{alba1}
  \text{l.h.s. of }\eqref{kokeshi}= \lim _{\e \da 0}  \left( -R_1(\e)+R_2(\e)
 \right)\,,
 \en 
  where 
  \begin{align*}
  & R_1(\e):=\int d \left[  \nu ^\e _{\tilde \o} -\nu ^\e _{\tilde \o, \D}\right](x,z) \nabla_\e \bar f_\e(x,z) \varphi (x) b ( \t_{x/\e} \tilde \o, z)\,,\\
& R_2(\e):=   \int d  \nu ^\e _{\tilde \o}(x,z)  \bar f_\e (x) \nabla_\e \varphi (x,z) \tilde b ( \t_{x/\e} \tilde \o, z)\,.
\end{align*}
  We claim that $\lim_{\e \da 0} R_1(\e)=0$.    
  We call $\ell $ the distance between  the support $\D_\varphi\subset \D $  of $\varphi$ and $\partial \D$.  Then in $R_1(\e)$ the contribution comes only from pairs $(x,z)$ such that $x\in \D_\varphi $ and $x+\e z \not \in S$ and therefore from pairs $(x,z)$ such that $x\in \D$ and $|z| \geq \ell/\e$:
  \be
   R_1(\e)=\int d \nu ^\e _{\tilde \o} (x,z) \nabla_\e \bar f_\e(x,z) \varphi (x) b ( \t_{x/\e} \tilde \o, z) \mathds{1}(x\in \D\,,\; |z| \geq \ell/\e)\,.
   \en
  By Schwarz inequality  we have therefore that 
 $
  R_1(\e)^2 \leq  I_1(\e) I_2(\e) $, where 
  \begin{align}
  & I_1(\e):=\int d \nu ^\e _{\tilde \o} (x,z) \nabla_\e \bar  f_\e(x,z)^2 | \varphi (x)|\mathds{1}(|z| \geq \ell/\e) \,,\\
  & I_2(\e):= \int d \nu ^\e _{\tilde \o} (x,z) |\varphi(x)| b ( \t_{x/\e} \tilde \o, z)^2 =\int d\mu^\e_{\tilde \o}(x)| \varphi(x)| \widehat{\, b^2\,}(\t_{x/\e}\tilde \o) \,.
  \end{align}
  Note that the last identity concerning $I_2(\e)$ uses that 
  $\tilde \o
 \in \O_{\rm typ}\subset \cA_1[ b^2]$. 
  Then $\lim_{\e\da 0} I_1(\e)=0$ due to Lemma \ref{lunghissimo}, while 
  $I_2(\e)$ converges to a bounded constant when $\e\da 0$ since  $\tilde \o \in \O_{\rm typ}\subset  \cA[ \widehat{\, b^2\,}]$. This proves that $R_1(\e)\to 0$.  

\smallskip  
    We now move to $R_2(\e)$. 
    \begin{Claim}\label{mengoni} We have 
    \be\label{charaba}
 \lim _{\e \da 0} \int d   \nu ^\e _{\tilde \o}(x,z) \Big{|} \bar f_\e (x) \bigl[ \nabla_\e \varphi (x,z)-\nabla \varphi (x) \cdot z\bigr] \tilde b ( \t_{x/\e} \tilde \o, z)\Big{|}=0\,.
    \en
    \end{Claim}
    \begin{proof}
     Given $\ell \in \bbN$ we write the integral in \eqref{charaba} as $A_\ell (\e)+B_\ell(\e)$, where $A_\ell(\e)$ is the contribution coming from $z$ with $|z|\leq \ell$ and $B_\ell(\e)$  is the contribution coming from $z$ with $|z|> \ell$.
Due to \eqref{mirra} and \eqref{re1} we can bound
\be
A_\ell(\e) \leq   C\ell^2  \e   \int  d  \nu ^\e _{\tilde \o}(x,z)  \bigl(  \phi (x) +\phi (x+\e z) \bigr) |\tilde b ( \t_{x/\e} \tilde \o, z)|\,.
\en
Hence, using now 
\eqref{micio2} in Lemma \ref{gattonaZ}, 
we can bound 
\be\label{ld1}
A_\ell (\e) \leq C\ell^2  \e   \int  d  \nu ^\e _{\tilde \o}(x,z) \phi (x) (|b|+ |\tilde b|)  ( \t_{x/\e} \tilde \o, z)\,.
\en
Since   $\o \in \O_{\rm typ}\subset \cA_1 [b]=  \cA_1 [|b|]$   (recall that $\tilde b\in \cW$ for all $b \in \cW$), the r.h.s. of \eqref{ld1} can be written as 
\be\label{ld2}
 C\ell ^2 \e   \int  d  \mu ^\e _{\tilde \o}(x) \phi (x) \left[ \widehat{ \, |b|\,}  +\widehat{ \, |\tilde b|\,}\right] ( \t_{x/\e} \tilde \o, z) 
 \en
 Since  $\o \in \O_{\rm typ}\subset \cA  [\widehat{\,| b|\, } ]\cap  \cA [\widehat{ \,  |\tilde b|\, }]$  (recall that $\tilde b\in \cW$ for all $b \in \cW$),
  the integral in  \eqref{ld2} converges to a finite constant as $\e\da 0$. Hence, coming back to \eqref{ld1}, $\lim _{\e \da 0} A_\ell (\e)=0$.

It remains to prove that $\lim _{\ell \uparrow \infty} \limsup_{\e\da 0} B_\ell (\e)=0$. We reason as above but now we apply  \eqref{paradiso} and a similar bound for $\nabla \varphi (x) \cdot z$. 
Due to \eqref{re1}, \eqref{paradiso}  and 
\eqref{micio2} in Lemma \ref{gattonaZ}, we can bound 
\be
B_\ell (\e) \leq C    \int  d  \nu ^\e _{\tilde \o}(x,z) \phi (x) (|b|+ |\tilde b|)  ( \t_{x/\e} \tilde \o, z) |z|\mathds{1}(|z|\geq \ell)\,.
\en
By Schwarz inequality
\be
B_\ell (\e) \leq C\,  C_\ell (\e)^{1/2} D_\ell (\e)^{1/2}
\en
where 
\begin{align*}
 C_\ell (\e):& =2\int  d  \nu ^\e _{\tilde \o}(x,z) \phi (x) (|b|^2 + |\tilde b|^2)  ( \t_{x/\e} \tilde \o, z) \\
 & = 2 \int  d  \mu ^\e _{\tilde \o}(x) \phi (x) \left( \widehat{|b|^2} +\widehat{ |\tilde b|^2}\right)  ( \t_{x/\e} \tilde \o)  \\
 D_\ell(\e):& = \int  d  \nu ^\e _{\tilde \o}(x,z) \phi (x)  |z|^2\mathds{1}(|z|\geq \ell)= \int  d  \mu ^\e _{\tilde \o}(x) \phi (x) \hat h_\ell (\t_{x/\e}\tilde \o)\,,
\end{align*}   
where $h_\ell (\o,z):= |z|^2 \mathds{1}(|z|\geq \ell)$. 
Note that in the identities concerning $C_\ell(\e)$ and $D_\ell(\e)$ we have used that
$\tilde \o \in \O_{\rm typ} \subset \cA_1[ b^2]\cap \cA_1[ \tilde b ^2]$ and $\tilde \o \in \O_{\rm typ}\subset \cA_1[|z|^2]\subset \cA_1[ h_\ell]$.
As $\tilde \o \in \O_{\rm typ}$, which is included in the sets 
$\cA_1[ |b|^2]$, $ \cA_1[ |\tilde b|^2] $, $\cA[ \widehat{|b|^2}]$, $\cA[\widehat{ |\tilde b|^2}]$, $\cA_1[ h_\ell]$ and $\cA[\hat h_\ell]$,
we get
\be 
\limsup_{\e \da 0} B_\ell(\e) \leq C\left[ \int dx\, m \phi(x)   \bbE_0[\widehat{|b|^2} +\widehat{ |\tilde b|^2}] \right]^{1/2}\bbE_0[\hat h_\ell]^{1/2}\,,
 \en
and the r.h.s. goes to zero as $\ell\to \infty$.
\end{proof}
     \smallskip
      We come back to \eqref{kokeshi}. By combining \eqref{alba1}, \eqref{charaba}  and the limit $R_1(\e)\to 0$, we conclude that 
     \be\label{alba2}
  \text{l.h.s. of }\eqref{kokeshi}= \lim _{\e \da 0} \int d   \nu ^\e _{\tilde \o}(x,z)  \bar f_\e (x) \nabla \varphi (x)\cdot z  \tilde b ( \t_{x/\e} \tilde \o, z)\,.
    \en
 Due to \eqref{alba2} and    since $\eta_{\tilde b }=- \eta_b$, to prove \eqref{kokeshi} we only need to show that 
      \begin{equation}\label{tramontoA}
 \lim _{\e \da 0} \int d  \nu ^\e _{\tilde \o}(x,z) \bar  f_\e (x) \nabla \varphi (x) \cdot z \tilde b ( \t_{x/\e} \tilde \o, z)=
 \int   dx \,m  v (x) \nabla \varphi(x) \cdot \eta_{\tilde b}\,.
 \end{equation}
 To this aim we observe that 
 \be\label{favorita1}
\int d  \nu^\e_{\tilde \o} (x,z) \bar f _\e (x) \partial _i \varphi (x) z_i \tilde{b} (\t_{x/\e} \tilde\o, z)=
\int d  \mu^\e_{\tilde \o}(x)  \bar f_\e (x) \partial _i \varphi (x) u_{\tilde b,i}( \t_{x/\e} \tilde \o)\,,
\en
where $u _{\tilde b,i }(\o):= \int d\hat \o (z)
c_{0,z}(\o) z_i \tilde{b} ( \o, z)$ (recall that  $\tilde \o\in \O_{\rm typ} \subset \cA_1[ \tilde b(\o,z) z_i]$). 
We claim that 
 \be\label{favorita2}
 \lim_{\e \da 0}
 \int d   \mu^\e_{\tilde \o}(x)\bar   f_\e (x) \partial _i \varphi (x) u_{\tilde b,i}( \t_{x/\e} \tilde \o)= \int _\D dx \,m  v (x) \partial_i \varphi (x) 
\bbE_0[ u _{\tilde b,i}]
\,.
 \en
Since the r.h.s. equals   $\int _\D dx \,m  v (x) \partial_i \varphi (x) 
 (\eta_{\tilde b} \cdot e_i)$,   our target \eqref{tramontoA} then would  follow as a byproduct of \eqref{favorita1} and \eqref{favorita2}. 
 It remains therefore to prove \eqref{favorita2}. 
 Given $M\in \bbN$ let $u _{\tilde b,i,M}:= | u_{\tilde b ,i}| \mathds{1}\bigl(| u_{\tilde b ,i}|\geq M\bigr)  $. 
 Due to   Prop.~\ref{prop_ergodico} (recall that $\tilde b\in \cW$ for any $b \in \cW$ and that 
 $\tilde \o\in \O_{\rm typ}\subset \cA[ u _{ b, i, M}]$ for all $b\in \cW$)
 \[ \lim _{\e \da 0}  \int d \mu^\e_{\tilde \o}(x)  | \partial_i \varphi (x)|  u _{\tilde b,i,M}  ( \t_{x/\e} \tilde \o)   =\int dx \, m | \partial_i \varphi (x)| \bbE_0[  u_{\tilde b,i,M}]\,.
 \]
 As $u_{\tilde b, i } \in L^1(\cP_0)$ we then get that 
 \be\label{colline}
\lim_{M\uparrow \infty}  \lim _{\e \da 0}  \int d  \mu^\e_{\tilde \o}(x) | \partial_i \varphi (x)|  u_{\tilde b,i,M}  ( \t_{x/\e} \tilde \o) = \lim _{M\uparrow \infty}  \int dx\,m | \partial_i \varphi (x)| \bbE_0[  u_{\tilde b,i,M}] =0\,.
 \en
Due to \eqref{re1} and \eqref{colline}, to get \eqref{favorita2} it is enough to show that 
  \be\label{favorita3}
  \begin{split}
\lim_{M\uparrow \infty}  \lim_{\e \da 0} &
 \int d  \mu^\e_{\tilde \o}(x)  \bar f_\e (x) \partial _i \varphi (x) u_{\tilde b,i }( \t_{x/\e} \tilde \o) \mathds{1}( |u_{\tilde b,i}( \t_{x/\e} \tilde \o)|\leq M)  \\
& =\int _\D dx \,m  v (x) \partial_i \varphi (x) 
\bbE[ u _{\tilde b,i}]
\,.
\end{split}
 \en
 Note that in \eqref{favorita3} we can replace
 $d  \mu^\e_{\tilde \o}(x)  \bar  f_\e (x) \partial _i \varphi (x)$ by $d  \mu^\e_{\tilde \o,\D}(x)   f_\e (x) \partial _i \varphi (x)$.
 Due to \eqref{re1} and since    $ u_{\tilde b,i} \mathds{1}( |u_{\tilde b,i}|\leq M) \in \cG$, by \eqref{rabarbaro} we have 
   \be\label{favorita4}
  \begin{split}
 \lim_{\e \da 0} &
 \int d \mu^\e_{\tilde \o,\D}(x)  f_\e (x) \partial _i \varphi (x) u_{\tilde b,i }( \t_{x/\e} \tilde \o) \mathds{1}( |u_{\tilde b,i}( \t_{x/\e} \tilde \o)|\leq M)  \\
& =\int _\D dx \,  m v (x) \partial_i \varphi (x) 
\bbE[ u _{\tilde b,i}  \mathds{1}( |u_{\tilde b,i}|\leq M)   ]
\,.
\end{split}
 \en
By dominated convergence, we get \eqref{favorita3} from \eqref{favorita4}.
 \end{proof}


\section{2-scale limit  points of $ V_\e$ and $\nabla_\e V_\e$}\label{limit_points}

In this section $\tilde \o$ is a fixed configuration  in  $ \O_{\rm typ}$.  Due to Lemmas \ref{paletta}, \ref{compatto1} and \ref{compatto2} along a subsequence $\e_k$ we have that  
\begin{align}
& L^2(\mu^\e_{\tilde \o, \L}) \ni V_\e \stackrel{2}{\toup} v  \in L^2(\L \times \O, m  dx \times \cP_0 ) \,,\label{totani1a}\\
& L^2(\nu^\e _{\tilde \o,\L})\ni   \nabla_\e V_\e \stackrel{2}{\toup} w\in   L^2(\L\times \O\times \bbR^d \,,\, m dx \times \nu ) \,, \label{totani2a}
\end{align}
for suitable functions $v$ and $w$.
In the rest of this section, when considering the limit $\e\da 0$, we understand that $\e$ varies in the sequence $\{\e_k\}$.

\begin{Proposition}\label{diamanti}  Let $v$  be as in \eqref{totani1a}. Then
   $v -\psi_{|\L}   \in H^1_0(\L, F,d_*)$.
\end{Proposition}
\begin{proof} 
We apply the results of Section \ref{sec_bike} to the case $\D=S$ and $f_\e:= V_\e -\psi$.
Since $f_\e$ is zero on $S\setminus \L$ and takes  values in $[-1,1]$ on $\L$, conditions \eqref{re1} and \eqref{re2} are satisfied.
In addition, we have $\nabla _\e f_\e (x,z)=0$ if $\{x,x+\e z\}$ does not intersect $\L$ and therefore  $\| f_\e\|_{L^2(\nu^\e_{\tilde \o, S})}= \| f_\e\|_{L^2(\nu^\e_{\tilde \o, \L})}
$.  By Lemma \ref{paletta}  we therefore conclude that also \eqref{re3} is satisfied.

At cost to refine the subsequence $\{\e_k\}$, without loss of generality we can assume that along $\{\e_k\}$ itself  we have 
\begin{align}
&L^2(  \mu^\e_{\tilde \o, S})\ni  f_\e \stackrel{2}{\toup} \hat v  \in L^2(S \times \O,  m dx \times \cP_0 ) \,,\label{totani101}\\
& L^2(\nu^\e_{\tilde\o,S})\ni \nabla_\e{f}_\e \stackrel{2}{\toup} \hat w\in   L^2(S\times \O\times \bbR^d  \,,\, m dx \times \nu ) \,, \label{totani202}
\end{align}
for suitable functions $\hat v,\hat w$. 
By Proposition \ref{prova2} we have $\hat v= \hat v(x)$.
\smallskip
We recall that in the proof of Proposition~\ref{oro} we have in particular derived \eqref{kokeshi}:
 for   each  solenoidal form $b  \in L^2_{\rm sol}(\nu) $ and each  function $\varphi \in C^2_c(S)$, it holds 
 \begin{equation}\label{kokeshi74}
 \int _S dx \,  \varphi(x) \int  d \nu (\o ,z) \hat w(x,\o,z) b(\o, z) =  
  -\int  _S dx\,   \hat v(x) \nabla \varphi(x) \cdot \eta_b\,.
\end{equation}
Since $f_\e \equiv 0$ on $S\setminus \L$, it is simple to derive from the definition of $2$--scale convergence  that $\hat v(x) \equiv 0$ $dx$--a.e.  on $S\setminus \L$ and that $\hat w(x,\cdot,\cdot)\equiv 0 $ $dx$--a.e. on $S\setminus \L$.  Therefore \eqref{kokeshi74} implies that 
\be\label{anatra}
\begin{split}
&  \Big |
\int  _{\L} dx\,  \hat  v(x) \nabla \varphi(x) \cdot \eta_b
\Big|
 = \Big|\int _\L dx \,  \varphi(x) \int  d \nu (\o ,z)\hat  w(x,\o,z) b(\o, z)\Big|\,.
 \end{split}
\en
By Schwarz inequality  we can bound
\begin{multline}
 C^2:=\int _\L dx \Big[ \int  d \nu (\o ,z) \hat w(x,\o,z) b(\o, z) \Big]^2  \\
 \leq 
  \int _\L dx  \int  d \nu (\o ,z) \hat  w(x,\o,z)  ^2 \int d \nu (\o,z) b(\o,z)^2\\
   = \| \hat w \| ^2 _{L^2( \L\times \O, dx\times \nu) }\|b \|^2_{L^2(\nu)}<+\infty\,.
 \end{multline}
 By applying now Schwarz inequality to \eqref{anatra} we conclude that 
 \be
  \Big |
\int  _{\L} dx\,   \hat v(x) \nabla \varphi(x) \cdot \eta_b
\Big | \leq C \| \varphi\|_{L^2(\L,dx)}\,.
\en
The above bound,   Proposition \ref{monti} and Corollary \ref{rock_bis} imply that $\hat v \in H^1_0(\L, F,d_*)$. To  get the thesis it remains to  observe that $\hat v= v-\psi_{|\L}$ $dx$--a.e. on $\L$, which follows from the definition of $2$--scale convergence, \eqref{totani1a}  and since $L^2(\mu^\e_{\tilde \o,\L}) \ni \psi _{\L}  \stackrel{2}{\rightharpoonup}  \psi_{\L} \in L^2(\L, dx)$.
\end{proof}

\begin{Proposition}\label{prova1} Let $w$  be as in \eqref{totani2a}. For $dx$--a.e. $x\in \L$, the map  $(\o,z) \mapsto w(x,\o,z)$ belongs to $L^2_{\rm sol}(\nu)$.
\end{Proposition}
\begin{proof} 
We use that $\la  \nabla_\e u,\nabla_\e V_\e \ra _{ L ^2( \nu ^\e_{\tilde\o,\L}) }=0$ for any 
$u\in H^{1,\e}_{\tilde \o, 0}$ (cf. Lemma \ref{benposto}--(ii)).  We take $u(x):=\e \varphi(x) g(\t_{x/\e} \tilde \o)$, where $\varphi \in C_c (\L)$ and  $g\in \cG_2$ (cf. Section \ref{sec_tipetto}).
Due to \eqref{leibniz} we have 
 \begin{equation}\label{aquilotto}
\nabla_\e  u (x,z)
= \e \nabla _\e \varphi (x,z) g(\t_{z+x/\e}  \tilde \o)+  \varphi(x) \nabla  g ( \t_{x/\e} \tilde  \o, z)\,,
\end{equation}
where $\nabla g(\o,z)= g(\t_z \o)-g(\o)$.
 Due to \eqref{aquilotto}, the identity $\la  \nabla_\e u,\nabla_\e V_\e \ra _{L^2 (\nu ^\e_{\tilde \o,\L} )}=0$ 
can be rewritten as 
\begin{equation}\label{cioccolata}
\begin{split}
& \e \int d  \nu ^\e_{ \tilde \o,\L} (x,z)  \nabla _\e \varphi (x,z) g(\t_{z+x/\e} \tilde \o) \nabla _\e V_\e (x,z) +\\
&  \int d  \nu ^\e_{ \tilde \o,\L}  (x,z)  \varphi(x) \nabla g ( \t_{x/\e}\tilde  \o, z)\nabla_\e  V_\e (x,z)  =0\,.
\end{split}
\end{equation}
We first show that 
\be\label{calmuccia}
\limsup_{\e \da 0} \Big| \int d  \nu ^\e_{ \tilde \o,\L} (x,z)  \nabla _\e \varphi (x,z) g(\t_{z+x/\e} \tilde \o) \nabla _\e V_\e (x,z)\Big| <+\infty\,.
\en
 By applying Schwarz inequality,
using that  $g$ is bounded as $g\in \cG_2$ and that  $\limsup_{\e\da 0}
\| \nabla_\e V_\e \| _{L^2 (\nu ^\e_{\tilde \o,\L} )}<+\infty $ due to \eqref{parasole} and since $\tilde \o \in \O_{\rm typ}\subset \O_2$, to get \eqref{calmuccia} it is enough to show that  $\limsup_{\e\da 0}
\| \nabla_\e \varphi    \| _{L^2 (\nu ^\e_{\tilde \o,\L} )}<+\infty $.  As $\tilde \o \in \O_{\rm typ}$,  by Lemma \ref{blocco}   it remains to prove that $\limsup_{\e\da 0}
\|\nabla \varphi (x)\cdot z  \| _{L^2 (\nu ^\e_{\tilde \o,\L} )}<+\infty $. To conclude we 
observe that, since  $\tilde \o \in \O_{\rm typ}\subset \cA_1[|z|^2]\cap \cA[\l_2]$, 
\begin{multline} \label{tiglio}
 \int d  \nu ^\e_{ \tilde \o,\L} (x,z)  |\nabla  \varphi (x)|^2 |z|^2\\
 =\int d  \mu ^\e_{ \tilde \o}(x)   |\nabla  \varphi (x)|^2  \l_2(\t_{x/\e}  \tilde \o)\to \int dx\, m   |\nabla  \varphi (x)|^2 \bbE_0[\l_2]
<+\infty\,.
\end{multline}
This completes the proof of \eqref{calmuccia}.

 Coming back to \eqref{cioccolata}, using \eqref{calmuccia} to treat the first addendum  and applying the 2-scale convergence  $\nabla_\e V_\e \stackrel{2}{\toup} w$ in \eqref{totani2a}   to  treat the second addendum,  we conclude that 
\be
\int _\L dx \int d\nu (\o,z) \varphi(x) \nabla g (\o, z) w (x,\o,z) =0 \qquad \forall g \in \cG_2\,.
\en
Note that above we have applied   \eqref{yelena} as  $\nabla g \in \cH_2\subset \cH$.
Since $\{\nabla g \,:\, g \in \cG_2\}$ is dense in $L^2_{\rm pot}(\nu)$,  the above identity implies that, for $dx$--a.e. $x\in \L$, the map  $(\o,z) \mapsto w(x,\o,z)$ belongs to $L^2_{\rm sol}(\nu)$. 
\end{proof}

\section{2-scale limit   of $ V_\e$: proof of  Theorem \ref{teo2} for $D_{1,1}>0$} \label{limitone}

In this section we  give the  proof of Theorem \ref{teo2} assuming that $D_{1,1}>0$. In particular, we will get \eqref{mortisia}.

\subsection{Convergence of $V_\e$ to $\psi$}\label{limitone_p1} We fix $\tilde \o \in \O_{\rm typ}$ and prove the convergences in Theorem \ref{teo2} for $\tilde \o$ instead of $\o$ there.
Due to Lemmas \ref{compatto1} and \ref{compatto2}
along a subsequence $\{\e_k\}$  we have
that  
$L^2(\mu^\e_{\tilde \o, \L})   \ni V_\e \stackrel{2}{\toup} v  \in L^2(\L \times \O,  m dx \times \cP_0 ) $ and $L^2(\nu^\e_{\tilde \o, \L})\ni \nabla_\e V_\e \stackrel{2}{\toup} w\in   L^2(\L\times \O\times \bbR^d\,,\, m dx \times \nu ) $ (cf. \eqref{totani1a} and \eqref{totani2a}). We claim that  for  $dx$--a.e. $x\in \L$ it holds 
\begin{equation}\label{mattacchione}
\int d \nu(\o, z)  w(x,\o, z)  z= 2 D \nabla_* v (x)\,.
\end{equation}
By Proposition~\ref{prova1} for $dx$--a.e. $x\in \L$, the map  $(\o,z) \mapsto w(x,\o,z)$ belongs to $L^2_{\rm sol}(\nu)$.
 On the other hand, by Proposition~\ref{oro} we know that 
$w(x,\o,z)= \nabla_*  v (x) \cdot z + v_1(x,\o,z)$, 
where $v_1\in L^2\bigl( \L, L^2_{\rm pot} (\nu)\bigr )$. Hence, by \eqref{jung}, for $dx$--a.e. $x\in \L$ we have that 
$v_1 (x,\cdot,\cdot ) = \mathbf{v}^a $, where  $ a:= \nabla_* v (x)$.
As a consequence (using also \eqref{solare789}),  for $dx$--a.e. $x\in \L$, we have
\begin{equation*}
\int d \nu(\o, z)  w(x,\o, z)  z= \int d \nu(\o,z) z [  \nabla_* v (x)\cdot z + \mathbf{v} ^{ \nabla_*  v (x)}(\o,z)]= 2 D \nabla_* v (x)\,,
\end{equation*}
thus proving \eqref{mattacchione}.

We now take a function  $\varphi \in C^2_c(\bbR^d)$ which is zero on $S\setminus \L$ (note that we are not taking $\varphi \in C^2_c (S)$). By Lemma \ref{benposto}--(ii) we have the identity $\la \nabla _\e \varphi, \nabla_\e V_\e\ra_{L^2(\nu^\e_{\tilde \o,\L})}=0$. The above identity and Lemma \ref{blocco} (use that $\tilde \o \in \O_{\rm typ}$)  imply that 
\be\label{george}
 0=\la \nabla _\e \varphi, \nabla_\e V_\e\ra_{L^2(\nu^\e_{\tilde \o,\L})}= \int d\nu^\e_{\tilde \o,\L} (x,z) \nabla \varphi(x) \cdot z \nabla_\e V_\e(x,z)+o(1)\,.
\en
Hence
\be\label{rino_gaetano}
0= \lim_{\e \da 0} \int d\nu^\e_{\tilde \o,\L} (x,z) \nabla \varphi(x) \cdot z \nabla_\e V_\e(x,z)\,.
\en
For each $n\geq 3$ let $A_n:=[-1/2+1/n,1/2-1/n]^d$ and let $\phi_n\in C_c( \L)$ be a function with values in $[0,1]$ such that $\phi_n\equiv 1$ on $A_n$.
By Schwarz inequality
\begin{multline}\label{manovra}
\Big| \int d\nu^\e_{\tilde \o,\L} (x,z) (\phi_n(x)-1)\nabla \varphi(x) \cdot z \nabla_\e V_\e(x,z)
\Big|\\
\leq \| \nabla_\e V_\e\|_{L^2(\nu^\e_{\tilde \o,\L})}
\Big[
\int _{\L \setminus A_n} d\mu^\e_{\tilde \o,\L} (x) \l_2 (\t_{x/\e} \tilde \o)\Big]^{1/2}\,.
\end{multline}
By ergodicity (equivalently by  applying Prop. \ref{prop_ergodico} to  suitable functions $\varphi, \varphi' \in C_c(\bbR^d)$ with $\varphi\leq \mathds{1}_{\L \setminus A_n }\leq \varphi'$ and using that $\tilde \o \in \O_{\rm typ} \subset \cA[\l_2]$)  we have 
$\lim_{\e\da 0} \int _{\L \setminus A_n} d\mu^\e_{\tilde \o}(x) \l_2 (\t_{x/\e} \tilde \o)= \ell(\L \setminus A_n )\bbE_0[\l_2]$.  As a byproduct with   Lemma \ref{paletta} we conclude that 
\be\label{manovra_bis}
\lim_{n\uparrow \infty} \limsup_{\e\da 0} \text{ l.h.s. of } \eqref{manovra}=0\,.
\en
Using \eqref{rino_gaetano} we get 
\be\label{kipur65}
\lim_{n \uparrow \infty} \limsup_{\e \da 0} \int d\nu^\e_{\tilde \o,\L} (x,z) \phi_n(x) \nabla \varphi(x) \cdot z \nabla_\e V_\e(x,z)=0\,.
\en

On the other hand,  due to \eqref{totani2a} and since  $\tilde \o \in \O_{\rm typ}$ (recall that the form $(\o,z) \mapsto z_i $ belongs to $\cH$, recall that $\phi_n \in C_c(\L)$ and  apply \eqref{yelena}), we can rewrite \eqref{kipur65} as 
\be
\lim_{n \uparrow \infty}   \int _\L dx \int d\nu (\o,z)  \phi_n(x) \nabla \varphi(x) \cdot z  w (x,\o,z)=0\,.
\en  Reasoning as in \eqref{manovra} we get

\be\label{flauto}
0= \int _\L dx \int d\nu (\o,z)  \nabla \varphi(x) \cdot z  w (x,\o,z)\,.
\en
As a byproduct of \eqref{mattacchione} and \eqref{flauto} we conclude that $0= \int_\L dx \nabla \varphi (x) \cdot D \nabla_* v(x)= \int_\L dx \nabla_* \varphi (x) \cdot D \nabla_* v(x) $
 for any 
 $\varphi \in  C^2_c(\bbR^d)$ with $\varphi\equiv 0$ on  $S\setminus \L$ (we write $\varphi \in \cC$). 
 If we take $ \varphi \in C^\infty _c(\bbR^d \setminus F)$, then $ \varphi_{|\L} $ can be approximated in the  space $H^1(\L$)  by functions  $\tilde \varphi_{|\L}$ with $\tilde \varphi \in \cC$. 
  Hence by density we conclude that $0= \int_\L dx \nabla_* \varphi (x) \cdot D\nabla_* v(x) $ for any $\varphi \in H_0^1(\L,F,d_*)$.
 Due to Proposition  \ref{diamanti} we  also have that $v\in K$ (cf. \eqref{kafka} in Definition~\ref{vettorino}).  
 Hence, by Definition \ref{fete} and Lemma \ref{unico},   $v$ is the unique weak solution   of the equation $ \nabla_* \cdot ( D \nabla_* v ) =0$ 
with boundary conditions  \eqref{mbc}.
By Corollary \ref{h2o}  we conclude    that $v=\psi_{|\L}$. Since the limit point is always $\psi_{|\L}$ whatever the subsequence $\{\e_k\}$, we get the $V_\e  \in L^2(\mu^\e_{\tilde \o, \L}) $ weakly 2-scale converges to  $ \psi_{|\L}  \in   L^2(\L \times \O, m  dx \times \cP_0 ) $ as $\e\da 0$, and not only along some subsequence. As $\psi_{|\L}$ does not depend from $\o$ and since $1\in \cG$, we derive from \eqref{rabarbaro} that $L^2(\mu^\e_{\tilde \o, \L})\ni V_\e \toup \psi \in L^2( \L, m dx)$ according to Definition \ref{debole_forte}.
\subsection{Convergence of the energy flow}\label{platano} Let us show that, given  
 $\tilde \o \in \O_{\rm typ}$, it holds
$\lim_{\e\da 0}  \frac{1}{2} \la \nabla_\e V_\e, \nabla _\e V_\e \ra 
_{L^2(\nu^\e_{\tilde \o, \L})}=m D_{1,1}\,.$
To this aim we  apply Lemma \ref{benposto}--(ii) with $u:= V_\e - \psi$, which  belongs to $H^{1,\e}_{0, \tilde\o}$.  Then we have 
$ \la \nabla_\e( V_\e - \psi), \nabla_\e V_\e \ra_{L^2(\nu^\e_{\tilde \o, \L})} =0$. This implies that 
\be \label{antiguaz}
\la \nabla_\e V_\e, \nabla_\e V_\e\ra_{L^2(\nu^\e_{\tilde \o, \L})}  = \la \nabla_\e \psi, \nabla_\e V_\e \ra_{L^2(\nu^\e_{\tilde \o, \L})}\,.
\en
\begin{Claim} \label{marlena93}
It holds $\lim _{\e\da 0} \int d \nu^\e_{\tilde \o, \L}(x,z) |\nabla_\e \psi(x,z)- z_1 |^2 =0$.
\end{Claim} 
\begin{proof} If $x, x+\e z\in \L$, then $\nabla_\e \psi(x,z)=z_1$.  We have only 4 relevant alternative cases: 
(a) 
$x\in \L$, $x+\e z\in S_+$; (b) $x\in S_+$, $x+\e z\in \L$; (c) $x\in \L$, $x+\e z\in S_-$; (d) $x\in S_-$, $x+\e z\in \L$.
Below we treat only case (a), since the other cases can be treated similarly. Hence we assume (a) to hold. Then $x_1+\frac{1}{2}=\psi(x) \leq \psi (x+\e z) \leq x_1+\e z_1+\frac{1}{2}$ and therefore $0\leq \nabla_\e \psi (x,z)  \leq z_1$. This implies that $|\nabla_\e \psi(x,z)- z_1 |^2\leq z_1^2$.
Fix $\d \in (0,1/2)$ and set $\L_\d=(-1/2+\d, 1/2-\d)^{d}$. 
We can bound
\be\label{mary1}
\begin{split}
 & \int d \nu^\e_{\tilde \o, \L}(x,z) |\nabla_\e \psi(x,z)- z_1 |^2 \mathds{1}( x\in \L_\d, x+\e z\in S_+
 )
\\
& \leq  \int  d \nu^\e_{\tilde \o}(x,z) z_1^2 \mathds{1}(x\in \L_\d , z_1 \geq \d /\e)\\
& \leq   \int_{\L_\d} d \mu^\e_{\tilde \o} (x)  \int d \widehat{\t_{x/\e} \tilde \o}(z)  c_{0,z}(\t_{x/\e}\tilde \o)z_1^2\mathds{1}(|z|\geq \d/\e)\\
&\leq \k(\d/\e) \int_{\L_\d} d \mu^\e_{\tilde \o} (x)  \int d \widehat{\t_{x/\e} \tilde \o}(z)  c_{0,z}(\t_{x/\e}\tilde \o)^\a z_1^2 \leq  \k(\d/\e) \int_{\L_\d} d \mu^\e_{\tilde \o} (x)  h (\t_{x/\e} \tilde \o)\,,
\end{split}
\en
where $\k (\ell) :=\sup_{\o \in \O_0, |z| \geq \ell} c_{0,z}(\o) ^{1-\a} $ and     
$h (\o) := \int d\hat \o(z) c_{0,z}(\o)^\a z_1^2$. We have that $\lim_{\e\da 0}\k(\d/\e) =0$ by \eqref{downtown}.
Since  $\o \in \O_{\rm typ}\subset \cA_1[c_{0,z}(\o) ^\a z_1^2]\cap \cA[ h]$, the last integral in \eqref{mary1}  converges to a finite constant as $\e\da 0$. This concludes the proof that the l.h.s. of \eqref{mary1} converges to zero as $\e \da 0$.

We can bound
\be \label{mary2}
\begin{split}
 & \int d \nu^\e_{\tilde \o, \L}(x,z) |\nabla_\e \psi(x,z)- z_1 |^2 \mathds{1}( x\in\L\setminus  \L_\d, x+\e z\in S_+
 )
\\
& \leq  \int d \nu^\e_{\tilde \o}(x,z) z_1^2 \mathds{1}( x\in \L\setminus \L_\d) \leq  \int _{\L\setminus \L_\d} d \mu^\e_{\tilde \o}(x) \l_2(\t_{x/\e}\tilde \o)\,.
\end{split}
\en
By Prop. \ref{prop_ergodico} and since  $\tilde \o \in \O_{\rm typ}\subset \cA_1[|z|^2]\cap \cA[\l_2]$,  $\lim_{\e \da 0}  \int _{\L\setminus \L_\d} d \mu^\e_{\tilde \o}(x) \l_2(\t_{x/\e}\tilde \o)= \ell(\L\setminus \L_\d) \bbE_0[\l_2]$. It then follows that the l.h.s. of \eqref{mary1} converges to zero as $\e \da 0$ and afterwards $\d \da 0$.
\end{proof}
As a  byproduct of Claim \ref{marlena93} and \eqref{antiguaz}, we get 
\be\label{antigua}
\lim_{\e\da 0}\la \nabla_\e V_\e, \nabla_\e V_\e\ra_{L^2(\nu^\e_{\tilde \o, \L})}  =   \lim_{\e \da 0} \int d \nu^\e_{\tilde \o,\L} (x,z) z_1 \nabla_\e V_\e (x,z)\,.
\en
By applying Schwarz inequality as in \eqref{manovra}, we get that \begin{multline}\label{antigua_bis}
\lim_{\e \da 0}  \int d \nu^\e_{\tilde \o,\L} (x,z) z_1 \nabla_\e V_\e (x,z)=
\lim_{n \uparrow \infty}
\lim_{\e \da 0} \int d \nu^\e_{\tilde \o,\L} (x,z) \phi_n(x) z_1 \nabla_\e V_\e (x,z)\,.
\end{multline}

By Lemma \ref{compatto2} from any vanishing  sequence $\{\e_k\}$  we can extract a 
sub-subsequence $\{\e_{k_n}\}$ such that 
$  \nabla_\e V_\e \stackrel{2}{\rightharpoonup}w   $ along the sub-subsequence as in \eqref{totani2a}. Since $\phi_n \in C_c(\L)$,  
as a byproduct of \eqref{antigua} and \eqref{antigua_bis} we obtain that 
\be\label{emoji}
\begin{split}
\lim_{\e\da 0}  \la \nabla_\e V_\e, \nabla_\e V_\e\ra _{L^2(\nu^\e_{\tilde \o, \L})} &  =   \lim _{n\uparrow \infty} \int _{\L} dx\,   m
\phi_n(x) \int d\nu  (x,z) z_1   w(x,\o,z)\\
 &=    \int _{\L} dx\,  m
 \int d\nu  (x,z) z_1  w(x,\o,z)
\end{split}
\en
along $\{\e_{k_n}\}$.
Due to \eqref{mattacchione} the last term equals $m \int_\L 2 (D \nabla v(x)) \cdot e_1 dx $. Since $v=\psi_{|\L}$ as derived  in the first part of the proof,    we get that $\nabla v(x)=e_1$. As a consequence, the last term of \eqref{emoji} equals $2m  D_{11}$, thus allowing to conclude the proof.

\appendix

\section{Proof of Equations \eqref{eq_ailo0} and  \eqref{eq_ailo}}\label{ailo}
For simplicity of notation we write $i_{x,y}$ instead of $i_{x,y}(\o)$.
It is  also convenient to set $A_0:=\hat \o \cap \L_\ell$, $A_{- 1}:= \{x \in \hat \o \cap S_\ell\,:\, x_1 
\leq -\ell/2\}$ and  $A_{ 1}:= \{x \in \hat \o \cap S_\ell\,:\, x_1 
\geq \ell/2\}$.

We start proving \eqref{eq_ailo0}. Due to definition \eqref{ide1} of $\s_\ell(\o)$ we can write the r.h.s. of \eqref{eq_ailo0} as 
\be\label{goccina}
\s_\ell(\o)- \sum _{x\in A_{-1}} \sum_{\substack{y\in A_0:\\ y_1\leq \g}}i_{x,y} + \sum 
_{\substack{x\in A_0:\\ x_1\leq \g}} \, \sum_{\substack{y\in A_0\cup A_1:\\ y_1> \g}}i_{x,y}\,.
\en
By antisymmetry  $- \sum _{x\in A_{-1}} \sum_{\substack{y\in A_0:\\ y_1\leq \g}}i_{x,y} =  \sum _{x\in A_{-1}}\sum_{\substack{y\in A_0:\\ y_1\leq \g}}  i_{y,x}$. Hence, \eqref{goccina} can be rewritten as 
\be\label{goccina1}
\s_\ell(\o)+ \sum_{\substack{x\in A_0:\\ x_1\leq \g}} \sum _{y\in A_{-1}} i_{x,y} +\sum 
_{\substack{x\in A_0:\\ x_1\leq \g}} \, \sum_{\substack{y\in A_0\cup A_1:\\ y_1> \g}}i_{x,y}\,.
\en
By antisymmetry  $
\sum_{\substack{x\in A_0:\\ x_1\leq \g}}  \sum_{\substack{y\in A_0:\\ y_1\leq \g}} i_{x,y}=0$. By adding this zero sum to \eqref{goccina} we get 
$\s_\ell(\o)+ \sum_{\substack{x\in A_0:\\ x_1\leq \g}}  ({\rm div}\, i)_x $,  $ ({\rm div}\, i)_x$ being  the divergence of the current field at $x$ given by 
 $({\rm div}\, i)_x :=\sum _{y \in \hat \o \cap S_\ell} i_{x,y}$. To conclude the proof of \eqref{eq_ailo0} we observe that $({\rm div}\, i)_x =0$  for any $x \in A_0$ by \eqref{thenero}.

\smallskip
We move to the proof of \eqref{eq_ailo}. Due to \eqref{ahahah} we can write the r.h.s. of \eqref{eq_ailo} as 
\begin{equation}\label{nebula}
 2^{-1} \sum_{ \substack{(x,y):\\ \{x,y\} \in \bbB^\o_\ell} } c_{x,y}(\o) \bigl( V^\o _\ell (x) - V^\o _\ell (y) \bigr)^2=2^{-1} C_1-2^{-1} C_2\,,
 \en
 where 
$ C_1:= \sum_{ \substack{(x,y):\\ \{x,y\} \in \bbB^\o_\ell} } i_{x,y}   V^\o _\ell (y)$ and 
  $C_2:=  \sum_{ \substack{(x,y):\\ \{x,y\} \in \bbB^\o_\ell} } i_{x,y}   V^\o _\ell (x)$.
 We analyze the two contributions $C_1$ and $C_2$ separately.
  As $V\equiv 0$ on $A_{-1}$ and $V\equiv 1$ on $A_1$ we can write
  \be\label{oasi}
  C_1=
   \sum_{ x\in A_{-1}, y \in A_0 } i_{x,y}   V^\o _\ell (y)+  \sum_{ x\in A_{0}, y \in A_0 } i_{x,y}  V^\o _\ell (y)+ \sum_{ x\in A_{0}, y \in A_1 } i_{x,y}+
    \sum_{ x\in A_{1}, y \in A_0 } i_{x,y} V^\o _\ell (y)
   \,.  \en
   Note that, by antisymmetry of the current, we can rewrite \eqref{oasi} as
   \be\label{oasi1}
   C_1= \sum_{ x\in A_{0}, y \in A_1 } i_{x,y}-  \sum_{ y \in A_0 } V^\o_\ell (y) \sum _{x\in A_0\cup A_{-1}\cup A_1} i_{y,x}=\sum_{ x\in A_{0}, y \in A_1 } i_{x,y}\,,
   \en
   where the last identity follows from the fact that $({\rm div}\, i)_x=0$ for any $x\in A_0$.
   
   We now move to $C_2$. Always by the above zero divergence property,  in $C_2$ we can remove the contribution from $x\in A_0$. Hence, using also \eqref{thebianco}, we get 
   \be\label{oasi2}
   C_2=  \sum_{ x\in A_{-1}, y\in A_0 } i_{x,y}  V^\o _\ell (x)+\sum_{ x\in A_{1}, y\in A_0 } i_{x,y}   V^\o _\ell (x)= \sum_{ x\in A_{1}, y\in A_0 } i_{x,y}  \,.
      \en
      By combining \eqref{nebula}, \eqref{oasi1} and \eqref{oasi2} we conclude that the r.h.s. of \eqref{eq_ailo} equals $ \sum_{ x\in A_{0}, y\in A_1 } i_{x,y}$. This last term equal $\s_\ell(\o)$ due to \eqref{eq_ailo0} with $\g$ very near to $\ell/2$ (as $\hat \o$ is a locally finite set).
\bigskip
\bigskip
\noindent
{\bf Acknowledgements}.
 I thank  Andrey Piatnitski  for useful discussions. 
I thank Annibale Faggionato and Bruna Tecchio for their warm hospitality in Codroipo, where part of this work has been completed. 


\end{document}